\documentclass[11pt]{article}
\usepackage{fullpage}

\usepackage[utf8]{inputenc} % allow utf-8 input
\usepackage[T1]{fontenc}    % use 8-bit T1 fonts
\usepackage{hyperref}       % hyperlinks
\usepackage{url}            % simple URL typesetting
\usepackage{booktabs}       % professional-quality tables
\usepackage{amsfonts}       % blackboard math symbols
\usepackage{nicefrac}       % compact symbols for 1/2, etc.
\usepackage{microtype}      % microtypography
\usepackage{xcolor}         % colors

\usepackage[style=alphabetic,backend=bibtex,maxnames=12,maxbibnames=10,maxcitenames=10,maxalphanames=10,giveninits=true,doi=false,url=true]{biblatex}
\newcommand*{\citet}[1]{\AtNextCite{\AtEachCitekey{\defcounter{maxnames}{2}}} \textcite{#1}}

\newcommand*{\citep}[1]{\cite{#1}}

\bibliography{clones2}

\usepackage{amsmath, amssymb, amsthm}
\usepackage{graphicx}
\usepackage[ruled,vlined]{algorithm2e}
\usepackage{subcaption}

\def\disteq{\stackrel{d}{=}}

\def\output{\mathcal{S}}
\def\lr{\mathcal{R}}

\def\shuffler{\mathcal{A}_{\rm s}}
\def\bin{{\rm Bin}}

\def\D{\mathcal{D}}
\def\Aldp[#1]{\mathcal{R}^{(#1)}}

\def\out[#1]{\mathcal{S}^{(#1)}}
\def\2RR{\texttt{2RR}}
\def\3RR{\texttt{3RR}}
\def\kRR{\texttt{kRR}}

\newcommand{\decompprob}[1]{p_{#1}}

\newcommand{\dalpha}[1]{D_{#1}}
\newcommand{\extremalset}[1]{\mathcal{E}_{{#1}}}

\newcommand{\zo}{\{0,1\}}
\newcommand{\zot}{\{0,1,2\}}

\providecommand{\equ}[1]{

\begin{equation}
#1
\end{equation}}

\providecommand{\alequn}[1]{\begin{align*} #1 \end{align*}}

\let\originalleft\left
\let\originalright\right
\renewcommand{\left}{\mathopen{}\mathclose\bgroup\originalleft}
\renewcommand{\right}{\aftergroup\egroup\originalright}

\newcommand{\ex}[2]{{\ifx&#1& \mathbb{E} \else \underset{#1}{\mathbb{E}} \fi \left[#2\right]}}
\newcommand{\pr}[2]{{\ifx&#1& \mathbb{P} \else \underset{#1}{\mathbb{P}} \fi \left[#2\right]}}

\providecommand{\E}{\mathop{\mathbb E}}

\providecommand{\Pr}{\mathop{\mathbb{Pr}}}

\newcommand{\eps}{{\varepsilon}}

\newcommand{\Ber}{\ensuremath{\mathrm{Bern}}}

\usepackage{times}
\newcommand{\cS}{\mathcal{S}}

\newtheorem{thm}{Theorem}[section]
\newtheorem{theorem}[thm]{Theorem}

\newtheorem{lemma}[thm]{Lemma}

\newtheorem{corollary}[thm]{Corollary}

	\theoremstyle{definition}

\newtheorem{definition}[thm]{Definition}

\newcommand{\Bin}{\operatorname{Bin}}

\newcommand{\unif}[1]{\mathcal{U}_{#1}}

\newcommand{\threedistqone}[3]{P_0\left(#1, #2, #3\right)}
\newcommand{\threedistqtwo}[3]{P_1\left(#1, #2, #3\right)}
\newcommand{\threedistqb}[3]{P_b\left(#1, #2, #3\right)}

\newcommand{\threedistone}[2]{P_0\left(#1,#2\right)}
\newcommand{\threedisttwo}[2]{P_1\left(#1,#2\right)}
\newcommand{\threedistb}[2]{P_b\left(#1,#2\right)}

\newcommand{\threedistonemain}[1]{P_0\left(#1\right)}
\newcommand{\threedisttwomain}[1]{P_1\left(#1\right)}

\title{Stronger Privacy Amplification by Shuffling for R\'enyi and Approximate  Differential Privacy}
\date{}
\author{
	Vitaly Feldman\\
	Apple\\
	\and {Audra McMillan}\\
	Apple\\
	\texttt{audra\textunderscore mcmillan@apple.com}\\
	\and Kunal Talwar\\
	Apple\\
	\texttt{ktalwar@apple.com}\\
}

\begin{document}

\maketitle

\begin{abstract}
The shuffle model of differential privacy has gained significant interest as an intermediate trust model between the standard local and central models \citep{ErlingssonFMRTT19,CheuSUZZ19}. A key result in this model is that randomly shuffling locally randomized data amplifies differential privacy guarantees. Such amplification implies substantially stronger privacy guarantees for systems in which data is contributed anonymously \citep{Bittau17}.

In this work, we improve the state of the art privacy amplification by shuffling results both theoretically and numerically.
Our first contribution is the first asymptotically optimal analysis of the R\'enyi differential privacy parameters for the shuffled outputs of LDP randomizers. Our second contribution is a new analysis of privacy amplification by shuffling. This analysis improves on the techniques of \citet{FeldmanMT:2020} and leads to tighter numerical bounds in all parameter settings.
\end{abstract}

\newpage

\section*{Errata}
The general upper bounds stated in the original version of this paper had an error in the proof. The error occurs in the proof of Lemma 3.5 when $p<\frac{1}{e^{\eps_0}+1}$. The error affects the theorems and corollaries that use Lemma 3.5 in their proof: Theorem 3.1, Theorem 3.2, Corollary 4.3 and Theorem 5.2.

In this errata we will outline the error, and give updated results:
\begin{itemize}
    \item The general upper bounds Theorem 3.1 and Theorem 3.2, previously stated for any sequence of adaptive local randomizers, hold for a restricted class of local randomizers that contains several important local randomizers. The numerical results in Figure 1, Figure 2, and Figure 3 also hold for this class of local randomizers.
    \item The key lemma showing that the privacy analysis can be reduced to comparing two multinomial distributions, Lemma 3.5, holds with a slightly different decomposition of the local randomizers. Given any specific local randomizer, such a decomposition is guaranteed to exist, and hence this provides a method for numerically computing a privacy amplification by shuffling bound for any local randomizer.
    \item The asymptotic upper bound on the privacy amplification by shuffling in terms of R\'enyi differential privacy (Corollary 4.3) still holds. The original proof depended on Lemma 3.5, but the same result holds with slightly worse constants using the results from \cite{FeldmanMT:2020}. 
    \item A slight variant on the upper bound for privacy amplification by shuffling for \texttt{kRR} (Theorem 5.2) holds. The new upper bound no longer exactly matches the lower bound given in Theorem 5.3, although we show that the two bounds are numerically close.
\end{itemize}

An outline of the error, as well as the updated results stated above can be found in Section~\ref{errata}. The affected theorems are flagged in the main body, although in order to maintain consistency with the original paper, the prose has not been changed. We conjecture that the general results stated in the original paper do hold, although this is left as an open problem in this errata.

\section{Introduction}
We consider privacy-preserving data analysis in the federated setting augmented with a shuffler. In this model each client sends a report of their data and these reports are then anonymized and randomly shuffled before being sent to the server. Systems based on this model were first proposed by \cite{Bittau17} as a simple way to improve the privacy of the user data.

The interest in this model was spurred by two works \citep{ErlingssonFMRTT19,CheuSUZZ19}  demonstrating that shuffling can provably amplify DP guarantees. Specifically, if each of $n$ clients randomizes their data with $\eps_0$ {\em local} DP then the shuffled reports satisfy $(\eps(\eps_0,\delta,n),\delta)$ DP for some $\eps(\eps_0,\delta,n) \ll \eps_0$ (when $n$ and $1/\delta$ are sufficiently large). This fact can also be used to analyze models augmented with an aggregator such as PRIO \citep{Prio}
(since the sum of real values provides even less information than shuffled values). In particular, privacy amplification by shuffling was used in Apple's and Google's Exposure Notification Privacy-preserving Analytics \citep{ENPA:2021}. Since then a number of works have studied privacy amplification by shuffling and the model augmented with a shuffler more generally (e.g.~\citep{Balle:2019,Ghazi:2019,GhaziPV19,Balle:2020,Balle2020,GMPV20,cheu2020limits,CheuSUZZ19,BalcerCheu,WangXDZHHLJ20,ErlingssonFMRSTT20,FeldmanMT:2020,girgis2020shuffled}).

The key to applications of privacy amplification by shuffling is bounding the resulting privacy parameter $\eps$ (as a function of $\eps_0$, $\delta$ and $n$), especially in the $\eps_0>1$ regime that is crucial for obtaining sufficiently accurate results in practice. A number of works addressed these bounds although, until recently, known bounds were asymptotically suboptimal \citep{ErlingssonFMRTT19,Balle:2019,Balle:2020} or only applied to the binary randomized response \citep{CheuSUZZ19}. For a summary of these results see \cite[Table 1]{FeldmanMT:2020}.

In a recent work, \citet{FeldmanMT:2020} give an asymptotically optimal analysis of privacy amplification by shuffling for $(\eps, \delta)$-DP. They show a more general result\footnote{In their result the inputs are shuffled {\em before} applying local randomizers. It is more general since it implies amplification when outputs of identical randomizers are shuffled. In addition, it allows adaptive choice of local randomizers which is necessary for analyzing iterative optimization algorithms such as stochastic gradient descent.} that running an adaptive sequence of arbitrary $\eps_0$-DP local randomizers on a uniformly random permutation of $n$ data items, yields an $(\eps, \delta)$-DP algorithm, where $\eps = O\left((1-e^{-\eps_0})\frac{\sqrt{e^{\eps_0}\ln(1/\delta)}}{\sqrt{n}}\right)$. Their result relies on a reduction from analysis of the privacy parameter for general adaptive protocols to analysis of divergence between a fixed pair of distributions on 3 values. In particular, they obtain a way to compute an upper bound on $\eps(\eps_0,\delta,n)$ numerically leading to numerical bounds that significantly improve on prior work.

One limitation of bounds on the approximate DP parameter $\eps(\eps_0,\delta,n)$ is that they are not well-suited for analysis of multi-step algorithms common in machine learning. Analysis of such algorithms requires composition whereas composition in terms of $(\eps,\delta)$ parameters typically leads to a $\sqrt{\ln(1/\delta)}$ overhead in the resulting bound. This overhead can be addressed by analyzing the privacy loss in terms of R\'enyi differential privacy (RDP) \cite{DworkR16,DLDP,mironov2017renyi,Bun:2016}. R\'enyi DP for order $\alpha$ upper bounds the privacy loss using R\'enyi divergence of order $\alpha$. Crucially, it leads to simple and relatively tight bounds for composition and can be easily converted to approximate DP.

Bounds on RDP parameters of privacy amplification by shuffling were first given by \citet{ErlingssonFMRTT19}. Their bound is asymptotically optimal for $\eps_0 < 1$, but for $\eps \geq 1$, their bound of $O( \alpha  e^{6\eps_0}/n)$ on RDP of order $\alpha$ is suboptimal. \citet{girgis2020shuffled} improved the bound to $O( \alpha  e^{2\eps_0}/n)$ (albeit in a rather limited range of $\alpha$) and also prove a lower bound of $\Omega( \alpha  e^{\eps_0}/n)$. Numerically, the strongest bounds on RDP are given in \citep{FeldmanMT:2020} who also demonstrate the advantages of RDP-based bounds for composition (of shuffled outputs). Another numerical approach to composition for shuffled outputs is given in \citep{koskela2021tight}. Their bounds rely on the Fourier accountant \citep{koskela2020computing} and the reduction from \citep{FeldmanMT:2020}.

\subsection{Our contribution}
We improve the existing bounds in two, largely independent, ways.

Our first contribution is an asymptotically tight upper bound on the R\'enyi DP of shuffled outputs of $n$ local randomizers. Specifically, we show that for $\alpha \leq c_0 \frac{n}{\eps_0 e^\eps_0}$ (for some fixed constant $c_0 >0$) and $\eps_0 >1$, the RDP parameter of order $\alpha$ is $O(\alpha \frac{e^{\eps_0}}{n})$. This improves on the results in \citep{girgis2020shuffled} both in terms of the bound and in terms of the range of $\alpha$ as they only prove their bound of $O( \alpha  e^{2\eps_0}/n)$ for $\alpha \leq c_1 (\frac{n}{e^{5\eps_0}})^{1/4}$. In particular, their bound can only be used for $\eps_0 \leq \ln(n)/5$, whereas our bound is non-trivial for $\eps_0 \leq \ln(n)-\ln\ln(n)$. Bounds on higher order $\alpha$'s are necessary for converting the RDP bounds to approximate DP bounds with relatively small $\delta$. We also note that for $\alpha > \frac{n \eps_0}{ e^{\eps_0}}$, $\alpha \frac{e^{\eps_0}}{n} > \eps_0$ and thus our bound applies to almost the entire range of $\alpha$ where the bound $O(\alpha \frac{e^{\eps_0}}{n})$ is non-trivial.

Our proof relies on a general conversion from truncated Gaussian tail bounds of the privacy loss random variable to a bound on R\'enyi privacy loss that might be useful in other contexts. Previously such conversion was only known for pure DP \citep{mironov2017renyi,Bun:2016}. We apply this general conversion to the asymptotically optimal bounds for approximate DP that we derive. We remark that for the purpose of obtaining asymptotically optimal bounds we can also apply this conversion to the bounds in \citep{FeldmanMT:2020}. See Section \ref{sec:rdp} for more details.

Our second contribution is a new, stronger analysis of privacy amplification by shuffling. We follow the basic approach from \citep{FeldmanMT:2020} which reduces the divergence between distribution on the outputs of an adaptive application of $n$ local $\eps_0$-DP randomizers on two datasets that differ in a single element to analysis of the divergence between a fixed pair of distributions on 3 values.
Informally, the reduction in \citep{FeldmanMT:2020} shows that an LDP randomizer on any input can be seen as producing the output of the randomizer on either of the inputs on which the datasets differ with some probability. Thus running the randomizer on inputs that are identical in both datasets can be seen as outputting a random number of samples from the output distribution of the randomizer on the elements on which the datasets differ (referred to as ``clones'').

At a high level we rely on more delicate analysis that instead of ``cloning'' the entire output distributions on differing elements only clones the part of those distributions where the distributions actually differ. This analysis improves the probability of producing a ``clone'' from $1/(2e^{\eps_0})$ to $1/(1+e^{\eps_0})$. For $\eps_0 > 1$ this leads to a roughly factor 2 improvement in the expected number of ``clones'' which translates to roughly factor $\sqrt{2}$ improvement in the resulting bound (or allowing a factor 2 more steps of an algorithm for the same overall privacy budget). For comparison, we note that the gap between the known numerical upper and lower bounds is typically less than a factor 2 and our improvement closes most of this gap (see Figure \ref{approxgraphs}).

Our reduction has the property that it can exploit additional structure in the local randomizer to give improved bounds. In particular, for the standard $k$-randomized response (or $k$-RR) randomizer our reduction leads to a tight bound. We note that \citet{FeldmanMT:2020} also give a reduction that can exploit the additional structure present in $k$-RR. However their analysis requires a separate reduction for this case and does not lead to a tight bound.

\section{Preliminaries}

Differential privacy (DP) is a stability notion for randomized algorithms. Intuitively, an algorithm is differentially private if the distribution on outputs doesn't change too much when a single individual changes their data. There are several ways to formalize the notion of closeness of distributions that are commonly used to define variants of DP. The most popular are the hockey-stick divergence, used to define $(\eps,\delta)$-differential privacy, and the R\'enyi divergence, used to define $(\rho(\alpha),\alpha)$-R\'enyi differential privacy (RDP).

\begin{definition}
The {\em hockey-stick} divergence between two random variables $P$ and $Q$ is defined by: \[\dalpha{e^{\eps}}(P\|Q) = \int \max\{0, P(x)-e^{\eps} Q(x)\} dx,\] where we use the notation $P$ and $Q$ to refer to both the random variables and their probability density functions. We say that $P$ and $Q$ are $(\eps, \delta)$-indistinguishable if $\max\{\dalpha{e^{\eps}}(P\|Q), \dalpha{e^{\eps}}(Q\|P)\}\le\delta$. \end{definition}

\begin{definition}[R\'enyi divergence]
For two random variables $P$ and $Q$, the R\'enyi divergence of $P$ and $Q$ of order $\alpha>1$ is \[D^{\alpha}(P\|Q) = \frac{1}{\alpha-1}\ln \E_{x\sim Q}\left[\left(\frac{P(x)}{Q(x)}\right)^{\alpha} \right] .\]
For $\alpha =1$, $D^{1}(P\|Q) = \mbox{KL}(P\|Q) = \E_{x\sim P}\left[\ln\left(\frac{P(x)}{Q(x)}\right)\right]$.
\end{definition}

The hockey-stick divergence and the R\'enyi divergence share several important properties that make them appropriate distance measures for measuring privacy. The data processing inequality is considered a hallmark of distance measures used for measuring privacy. It states that the privacy guarantee can not be degraded by further analysis of the output of a private mechanism. All the commonly used notions of privacy satisfy the post-processing inequality, including $D_{e^{\eps}}$ and $D^{\alpha}$.

\begin{definition}\label{postprocess} A distance measure $D:\Delta(\mathcal{S})\times\Delta(\mathcal{S})\to[0,\infty]$ on the space of probability distributions satisfies the data processing inequality if for all distributions $P$ and $Q$ in $\Delta(\mathcal{S})$ and (possibly randomized) functions $f:\mathcal{S}\to\mathcal{S'}$, \[D(f(P)\|f(Q))\le D(P\|Q).\]
\end{definition}

There are also several differential trust models of DP. We will be primarily interested in the shuffle model, but let us first introduce the more common central model and local model.
In the central model \citep{Dwork:2006}, the data of the individuals is held by the curator. The curator is trusted to analyse the data and enforce the privacy constraint.
In the local model, formally introduced in \cite{Kasiviswanathan:2008}, each individual (or client) randomizes their data before sending it to data curator (or server). This means that individuals are not required to trust the curator. The central model requires a high level of trust, but allows for significantly more accurate algorithms.
Since it requires less trust, most deployment of DP in industry use the local DP model \citep{Erlingsson:2014, Apple2017, Bolin:2017, ErlingssonFMRSTT20}.

We say that two databases are neighboring if they differ on the data of a single individual. We'll define the trust models with respect to the hockey-stick divergence, but note that the definitions for R\'enyi DP only differ in the choice of distance measure.

\begin{definition}[Central DP]
An algorithm $\mathcal{A}:\mathcal{D}^n\to\output$ is $(\eps, \delta)$-\emph{differentially private} if for all neighboring databases $X$ and $X'$, $\mathcal{A}(X)$ and $\mathcal{A}(X')$ are $(\eps, \delta)$-indistinguishable.
\end{definition}

In local DP, users interact with the server and send outputs of randomized algorithm. In the fully adaptive case, they can communicate with the server in an arbitrary order with adaptive interaction. Formally, a protocol satisfies local $(\eps,\delta)$-DP if the transcripts of the interaction on any two pairs of neighbouring datasets are $(\eps,\delta)$-indistinguishable. In this paper, we will only be considering the adaptive, single round model, where each user sends a single report to the server. In this setting, the condition on transcripts reduces to each user interacting with the server using a mechanism that is $(\epsilon,\delta)$-differentially private with respect that that users data. We call such mechanisms for the local reports of a user \emph{local randomizers.}

\begin{definition}[Local randomizer]\label{localrandomizer}
An algorithm $\lr\colon \D\to \cS$ is $(\eps, \delta)$-DP \emph{local randomizer}  if
for all pairs $x,x'\in \D$, $\lr(x)$ and $\lr(x')$ are $(\eps, \delta)$-indistinguishable.
\end{definition}

Formally, an adaptive single pass $(\eps,\delta)$-DP local protocol can be described by a sequence of local randomizers $\Aldp[i]:\out[1]\times\cdots\times\out[i-1]\times\mathcal{D}\to\out[i]$ for $i\in[n]$, where $\D$ is the data domain, $\out[i]$ is the range space of $\Aldp[i]$ and the $i$-th user returns $z_i=\Aldp[i](z_{1:i-1}, x_i)$. We require that the local randomizer $\Aldp[i](z_{1:i-1}, \cdot)$ be $(\eps,\delta)$-DP for all values of auxiliary inputs $z_{1:i-1}\in\out[1]\times\cdots\times\out[i-1]$. 

In all models, if $\delta=0$ then we will refer to an algorithm as $\eps$-differentially private. Further, we will occasionally refer to $\delta=0$ as pure differentially private and $\delta>0$ as approximate DP.

\section{Stronger Analysis of Privacy Amplification by Shuffling}\label{generalamplificationreduction}

In this section, we present a new reduction from analyzing the privacy guarantee of shuffling an adaptive series of local randomizers to analyzing the privacy guarantee of shuffling the output of a simple non-adaptive local algorithm with three outputs. Our reduction improves upon that of \cite{FeldmanMT:2020}, resulting in tighter numerical bounds for privacy amplification by shuffling for both approximate and R\'enyi DP. Specifically, we show that there exists two families of multinomial distributions $\threedistone{\eps_0}{p}$ and $\threedisttwo{\eps_0}{p}$ such that for any two neighbouring datasets $X_0$ and $X_1$, and any adaptive series of $\eps_0$-local randomizers, there exists a post-processing function $f$ and $p\in[0,1/(e^{\eps_0}+1)]$ such that $f(\threedistone{\eps_0}{p})$ and $f(\threedisttwo{\eps_0}{p})$ are identically distributed to the output of the shuffled mechanism on $X_0$ and $X_1$, respectively. As a result, the privacy loss of the general adaptive setting of shuffling is no worse than the divergence between $\threedistone{\eps_0}{p}$ and $\threedisttwo{\eps_0}{p}$.

Let use begin by formally defining the distributions $\threedistone{\eps_0}{p}$ and $\threedisttwo{\eps_0}{p}$.
For any $p\in[0,1/(e^{\eps_0}+1)]$, define random variables $Y_p$, $Y_{1,p}^0$ and $Y_{1,p}^1$ as follows
\begin{equation}\label{individualreports}
Y_p = \begin{cases} 0 & \text{w.p.   } p\\ 1 & \text{w.p.   } p \\ 2 & \text{w.p.   } 1-2p \end{cases},\;\;\;\; Y_{1,p}^0 = \begin{cases} 0 & \text{w.p.   } e^{\eps_0}p\\ 1 & \text{w.p.   } p \\ 2 & \text{w.p.   } 1-e^{\eps_0}p-p \end{cases}\text{, and }\;\;\;\;Y_{1,p}^1 = \begin{cases} 0 & \text{w.p.   } p\\ 1 & \text{w.p.   } e^{\eps_0}p \\ 2 & \text{w.p.   } 1-e^{\eps_0}p-p \end{cases}
\end{equation}
For $b\in\{0,1\}$, to obtain a sample from $\threedistb{\eps_0}{p}$, sample one copy from $Y_{1,p}^b$ and $n-1$ copies of $Y_p$, the output $(n_0,n_1)$ where $n_0$ is the total number of 0s and $n_1$ is the total number of 1s. Equivalently, let $C\sim \bin(n-1,2p)$, $A\sim\bin(C, 1/2)$ and $\Delta_1 \sim \Ber(e^{\eps_0}p)$ and $\Delta_2 \sim \Bin(1-\Delta_1,
p/(1-e^{\eps_0}p))$, where $\Ber(q)$ denotes a Bernoulli random variable with bias $q$. Let
\begin{equation}\label{componentdistributions}
\threedistone{\eps_0}{p}=(A+\Delta_1, C-A+\Delta_2)\;\;\; \text{and} \;\;\;\threedisttwo{\eps_0}{p}=(A+\Delta_2, C-A+\Delta_1).
\end{equation}

The following theorem is our improved general upper bound.

\begin{theorem}\label{mainanalyticalthm}
\textcolor{magenta}{Errata: The error in the proof of Lemma 3.5 affects this theorem. It holds for a restricted class of local randomizers (see Theorem~\ref{newtheorem31}). The replacement for Lemma 3.5 (Lemma~\ref{shuffletobinoms2new}) implies a  bound that depends on a particular local randomizer.} For a domain $\mathcal{D}$, let $\Aldp[i]:\out[1]\times\cdots\times\out[i-1]\times\mathcal{D}\to\out[i]$ for $i\in[n]$ (where $\out[i]$ is the range space of $\Aldp[i]$) be a sequence of algorithms such that $\Aldp[i](z_{1:i-1}, \cdot)$ is an $\eps_0$-DP local randomizer for all values of auxiliary inputs $z_{1:i-1}\in\out[1]\times\cdots\times\out[i-1]$. Let $\shuffler:\mathcal{D}^n\to \out[1]\times\cdots\times \out[n]$ be the algorithm that given a dataset $x_{1:n}\in\mathcal{D}^n$, samples a permutation $\pi$ uniformly at random, then sequentially computes $z_i=\Aldp[i](z_{1:i-1}, x_{\pi(i)})$ for $i\in[n]$ and outputs $z_{1:n}$. Let $X_0$ and $X_1$ be two arbitrary neighboring datasets in $\mathcal{D}^n$. Then for any distance measure $D$ that satisfies the data processing inequality,
\[D(\shuffler(X_0)\|\shuffler(X_1))\le D\left(\threedistone{\eps_0}{\frac{1}{e^{\epsilon_0}+1}}\Big\|\threedisttwo{\eps_0}{\frac{1}{e^{\epsilon_0}+1}}\right).\]
\end{theorem}

Theorem~\ref{mainanalyticalthm} is a strict improvement over \cite[Theorem 3.2]{FeldmanMT:2020}, the best known upper bound from prior work for adaptive shuffling using arbitrary $\eps_0$-DP local randomizers.
When $p=1/(e^{\eps_0}+1)$, the formula for $\threedistb{\eps_0}{p}$ simplifies to $\threedistone{\eps_0}{p} = (A+\Delta, C-A+1-\Delta)$ and $\threedisttwo{\eps_0}{p} = (A+1-\Delta, C-A+\Delta)$ where $C\sim \bin(n-1,2/(e^{\eps_0}+1))$, $A\sim\bin(C, 1/2)$ and $\Delta \sim \Ber(e^{\eps_0}/(e^{\eps_0}+1))$.
\cite[Theorem 3.2]{FeldmanMT:2020} is also a reduction to the divergence of a specific pair of local randomizers: let $C'\sim \bin(n-1,e^{-\eps_0})$, $A'\sim\bin(C', 1/2)$ and $\Delta' \sim \Ber(e^{\eps_0}/(e^{\eps_0}+1))$ then $Q_0(\eps_0) = (A'+\Delta', C'-A'+1-\Delta')$ and $Q_1(\eps_0) = (A'+1-\Delta', C'-A'+\Delta')$. Given the similarity of these two pairs of distributions, it is easy to show that there exists a post-processing function $g$ such that $g(Q_b(\eps_0))=\threedistb{\eps_0}{p}$. The post-processing inequality then implies that Theorem~\ref{mainanalyticalthm} is strictly tighter than \cite[Theorem 3.2]{FeldmanMT:2020}.

The similarity between $Q_b(\eps_0)$ and $\threedistb{\eps_0}{1/(e^{\eps_0}+1)}$ means that we can immediately obtain an analytic bound that improves on \citep{FeldmanMT:2020} by a factor of about $\sqrt{2}$.

\begin{theorem}\label{thm:shuffling4adp}
\textcolor{magenta}{Errata: The error in the proof of Lemma 3.5 affects this theorem. It holds for a restricted class of local randomizers. This theorem with slightly worse constants appears in \cite{Feldman:2022focs}.} For any domain $\mathcal{D}$, let $\Aldp[i]:\out[1]\times\cdots\times\out[i-1]\times\mathcal{D}\to\out[i]$ for $i\in[n]$ (where $\out[i]$ is the range space of $\Aldp[i]$) be a sequence of algorithms such that $\Aldp[i](z_{1:i-1}, \cdot)$ is an $\eps_0$-DP local randomizer for all values of auxiliary inputs $z_{1:i-1}\in\out[1]\times\cdots\times\out[i-1]$.
 Let $\shuffler:\mathcal{D}^n\to\out[1]\times\cdots\times \out[n]$ be the algorithm that given a dataset $x_{1:n}\in\mathcal{D}^n$, samples a uniform random permutation $\pi$ over $[n]$, then sequentially computes $z_i=\Aldp[i](z_{1:i-1}, x_{\pi(i)})$ for $i\in[n]$ and outputs $z_{1:n}$.
 Then for any $\delta\in[0,1]$ such that $\eps_0\le\ln(\frac{n}{8\ln(2/\delta)}-1)$, $\shuffler$ is $(\eps, \delta)$-DP,  where
\begin{equation}\label{epsbound}
\eps\le \ln\left(1+(e^{\eps_0}-1)\left(\frac{4\sqrt{2\ln(4/\delta)}}{\sqrt{(e^{\eps_0}+1)n}}+\frac{4}{n}\right)\right)
\end{equation}
\end{theorem}

 The proof of Theorem~\ref{mainanalyticalthm} relies on the following lemma that converts any $\eps_0$-DP local randomizer to one where for each output, the probability of the output takes one of two extremal values.

\begin{lemma}\label{sizetwopdf}\cite[Lemma IV.4]{ye2018optimal}
If $\mathcal{R}:\mathcal{D}\to\mathcal{S}$ is an $\eps$-DP local randomizer, and both $\mathcal{D}$ and $\mathcal{S}$ are finite, then there exists a finite output space $\mathcal{Z}$, a local randomizer $\mathcal{R}':\mathcal{D}\to\mathcal{Z}$, and a post-processing function $\Phi:\mathcal{Z}\to\mathcal{S}$ such that for all $z\in\mathcal{S}$, there exists $p_z\in[0,1/(e^{\eps_0}+1)]$ such that for all $x\in\mathcal{D}$, $\Pr(\mathcal{R}'(x)=z)\in\{p_z, e^{\eps_0}p_z\}$, and $\Phi\circ\mathcal{R'}=\mathcal{R}$.
\end{lemma}

The proof of Lemma~\ref{sizetwopdf} is contained in the proof of Lemma IV.4 in \cite{ye2018optimal}. For clarity we include a proof in Appendix~\ref{appendixreduction}. A special case of such a lemma was also proved in \cite{KOV14}. This allows us to prove the following corollary expressing randomizer outputs as mixtures of base distributions.

\begin{corollary}\label{decomposition}
Given any $\eps_0$-DP local randomizer $\mathcal{R}:\mathcal{D}\to\mathcal{S}$, and any $n+1$ inputs $x_1^0, x_1^1, x_2, \cdots, x_n \in\mathcal{D}$, if $\mathcal{S}$ is finite then there exists $p\in[0,1/(e^{\eps_0}+1)]$ and distributions $\mathcal{Q}_1^0$, $\mathcal{Q}_1^1$, $\mathcal{Q}_1, \mathcal{Q}_2, \cdots, \mathcal{Q}_n$ such that
\begin{align}
\nonumber\mathcal{R}(x_1^0) &= e^{\eps}p\mathcal{Q}_1^0+p\mathcal{Q}_1^1+(1-p-e^{\eps_0}p)\mathcal{Q}_1,\\
\nonumber\mathcal{R}(x_1^1) &= p\mathcal{Q}_1^0+e^{\eps_0}p\mathcal{Q}_1^1+(1-p-e^{\eps_0}p)\mathcal{Q}_1\\
\forall i\in[2,n], \;\;\mathcal{R}(x_i) &= p\mathcal{Q}_1^0+p\mathcal{Q}_1^1+(1-2p)\mathcal{Q}_i. \label{eq:decomposition}
\end{align}
\end{corollary}
\begin{proof}
Note that restricted to inputs $\{x_1^0,x_1^1,x_2,\cdots,x_n\}$, $\mathcal{R}$ satisfies the constraints of Lemma~\ref{sizetwopdf} so there exists an $\eps_0$-DP local randomizer $\mathcal{R}':\mathcal{D}\to \mathcal{Z}$, and post-processing function $\Phi$ such that for $z\in\mathcal{Z}$, there exists $p_z\in[0,1/(e^{\eps_0}+1)]$ such that for all $x\in\{x_1^0,x_1^1,x_2,\cdots,x_n\}$ $\Phi(\mathcal{R'}(x))=\mathcal{R}(x)$, and $\Pr(\mathcal{R}'(x)=z)\in\{e^{\eps_0}p_z,p_z\}$.

Let $L=\{z\in\mathcal{Z}\:|\; \Pr(\mathcal{R'}(x_1^0)=z)=e^{\eps_0}p_z
\text{ and } \Pr(\mathcal{R'}(x_1^1)=z)=p_z\}$ and $U=\{z\in\mathcal{S}\:|\; \Pr(\mathcal{R'}(x_1^0)=z)=p_z \text{ and } \Pr(\mathcal{R'}(x_1^1)=z)=e^{\eps_0}p_z\}$.
Let $M=\mathcal{Z}\backslash (L\cup U)$ and $p=\sum_{z\in L} p_z = \sum_{z\in U} p_z.$ Note that conditioned on the output lying in $L$, the distributions $\mathcal{R'}(x_1^0)$ and $\mathcal{R'}(x_1^1)$ are the same. Let $\mathcal{W}_1^0=\mathcal{R'}(x_1^0)|_L=\mathcal{R'}(x_1^1)|_L$. Similarly, let $\mathcal{W}_1^1=\mathcal{R'}(x_1^0)|_U=\mathcal{R'}(x_1^0)|_U$ and $\mathcal{W}_1 = \mathcal{R'}(x_1^0)|_{M}=\mathcal{R'}(x_1^1)|_{M}$.Then,
\begin{align*}
\mathcal{R}'(x_1^0) &= e^{\eps_0}p \mathcal{W}_1^0+p \mathcal{W}_1^1 + (1-e^{\eps_0}p-p)\mathcal{W}_1 \\
\mathcal{R}'(x_1^1) &= p \mathcal{W}_1^0+e^{\eps_0}p \mathcal{W}_1^1 + (1-e^{\eps_0}p-p)\mathcal{W}_1
\end{align*}
Further, for all $x_i\in\{x_2,\cdots, x_n\}$, $\mathcal{R}'(x_i)\ge p\mathcal{W}_1^0+p\mathcal{W}_1^1$, so there exists $\mathcal{W}_i$ such that \[\mathcal{R}'(x_i)=p \mathcal{W}_1^0+p \mathcal{W}_1^1 + (1-2p)\mathcal{W}_i.\]
Letting $\mathcal{Q}_1^0=\Phi(\mathcal{W}_1^0)$, $\mathcal{Q}_1^1=\Phi(\mathcal{W}_1^1)$, $\mathcal{Q}_1=\Phi(\mathcal{W}_1)$ and for all $i\in\{2,\cdots,n\}$, $\mathcal{Q}_i=\Phi(\mathcal{W}_i)$, we are done. Since we must have $p+e^{\eps_0}p\le 1$, $p\in[0,1/(e^{\eps_0}+1)].$
\end{proof}
We will use $\decompprob{\mathcal{R}}$ to denote the smallest value of $p$ for which $\mathcal{R}$ can be decomposed as in eqns~\eqref{eq:decomposition}, when $X_0$ and $X_1$ are clear from context. Intuitively, the smaller $\decompprob{\mathcal{R}}$ is, the closer $\mathcal{R}(x_1^0)$ and $\mathcal{R}(x_1^1)$ are, and the more amplification we'll see when using this local randomizer. We'll make this formal in Lemma~\ref{maxdivergence}.

\begin{lemma}\label{shuffletobinoms}
\textcolor{magenta}{Errata: there is an error in the proof of this lemma when $p<1/(e^{\eps_0}+1)$. An alternative version appears in Lemma~\ref{shuffletobinoms2new}.} For a domain $\mathcal{D}$, let $\Aldp[i]:\out[1]\times\cdots\times\out[i-1]\times\mathcal{D}\to\out[i]$ for $i\in[n]$ (where $\out[i]$ is the range space of $\Aldp[i]$, and $\out[i]$ is finite for all $i$) be a sequence of algorithms such that $\Aldp[i](z_{1:i-1}, \cdot)$ is an $\eps_0$-DP local randomizer for all values of auxiliary inputs $z_{1:i-1}\in\out[1]\times\cdots\times\out[i-1]$. Let $\shuffler:\mathcal{D}^n\to \out[1]\times\cdots\times \out[n]$ be the algorithm that given a dataset $x_{1:n}\in\mathcal{D}^n$, samples a permutation $\pi$ uniformly at random, then sequentially computes $z_i=\Aldp[i](z_{1:i-1}, x_{\pi(i)})$ for $i\in[n]$ and outputs $z_{1:n}$. Let $X_0$ and $X_1$ be two arbitrary neighboring datasets in $\D^n$ and $p^*\in[0,1/(e^{\eps_0}+1)]$ be such that with respect to $X_0$ and $X_1$, $\decompprob{\Aldp[i](z_{1:i-1})}\le p^*$ for all $i\in[n]$ and $z_{1:i-1}$. Then there exists a post-processing function $f$ such that $\shuffler(X_0)$ is distributed identically to $f(\threedistone{\eps_0}{p^*})$ and $\shuffler(X_1)$ is distributed identically to  $f(\threedisttwo{\eps_0}{p^*})$.
\end{lemma}

\begin{proof}

\begin{algorithm}[!t]
  \textbf{Input:} $x^0_1,x^1_1,x_2,\ldots, x_n$; $(n_0, n_1)\in\mathbb{N}^2$ \\
  Sample $y \in \zot^n$ to be a random permutation of a vector with $n_0$ 0s, and $n_1$ 1s and $n-n_0-n_1$ 2s.\\
  $J := \emptyset$, \\
  Let $K = \{i\in[n]\;|\; y_i=2\}$, $n_K=|K|$\\
  Let $\alpha=\frac{2p n_K(1-e^{\eps_0}p-p)}{2pn_K(1-e^{\eps_0}p-p)+(e^{\eps_0}p+p)(1-2p)(n-n_K)}$\label{alphaformula}\\
  Let $\chi=\Ber(\alpha)$\\
  \If{$\chi=1$}{
  Let $i$ be randomly and uniformly chosen from $K$\\
  $j_i=1$\\
  $K=K\backslash\{i\}$\\}
  \For{$i\in K$}{
  Let $j_i$ be a randomly and uniformly chosen element of $[2:n]\setminus J$\\
  $J := J \cup \{j_i\}$
  }
  \For{$t\in[n]$}{
  $y_t' = g(y_t, \decompprob{\Aldp[t](z_{1:t-1})})$\\
  Sample $z_t$ from $ \begin{cases} {\mathcal{Q}_1^0}^{(t)}(z_{1:t-1}) & \text{if } y_i'=0; \\
                                    {\mathcal{Q}_1^1}^{(t)}(z_{1:t-1}) & \text{if } y_i'=1;  \\
                                    \mathcal{Q}_{j_t}^{(t)}(z_{1:t-1}) & \text{if } y_i'=2. \end{cases}$\\
  }
  \textbf{return} $z_1, \ldots, z_n$
  \caption{Post-processing function, $f$} \label{postprocessing}
\end{algorithm}

Let $X_0=\{x_1^0,x_2,\cdots,x_n\}$ and $X_1=\{x_1^1,x_2,\cdots,x_n\}$ be two neighbouring datasets in $\mathcal{D}^n$.
As in the proof of \cite[Lemma 3.3]{FeldmanMT:2020}, the proof relies on a decomposition of the algorithm that shuffles the data and then applies the local randomizers, to an algorithm in which each client first reports which component of the mixture it will sample from, then applies shuffling to these reports and finally applies a post-processing step in which randomizers are applied according to the shuffled mixture component indices. The key difference between the proofs is that the behaviour of user 1 is more complicated, resulting in a more complicated post-processing function.

A description of the post-processing function is given in Algorithm~\ref{postprocessing}. We claim that for $b\in\{0,1\}$, $f(\threedistb{\eps_0}{p^*})\disteq \shuffler(X_b)$. The key observation is that by assumption we have decompositions such that for all $t\in[n]$, $z_{1:t-1}$, there exists $p_t \stackrel{de f}{=}\decompprob{\Aldp[t](z_{1:t-1})}\in[0,p^*]$ such that:
\begin{align*}
\forall b\in\{0,1\}, \;\;\mathcal{R}^{(t)}(z_{1:t-1}, x_1^b) &= e^{\eps}p_t{\mathcal{Q}_1^b}^{(t)}(z_{1:t-1})+p_t{\mathcal{Q}_1^{1-b}}^{(t)}(z_{1:i-1}) + (1-p_t-e^{\eps}p_t)\mathcal{Q}_1^{(t)}(z_{1:i-1}),
\end{align*}
\begin{align*}\forall i\in[2,n], \;\;\mathcal{R}^{(t)}(z_{1:t-1}, x_i) &= p_t{\mathcal{Q}_1^0}^{(t)}(z_{1:i-1})+p_t{\mathcal{Q}_1^1}^{(t)}(z_{1:i-1})
+(1-2p_t)\mathcal{Q}_i^{(t)}(z_{1:i-1}).
\end{align*}

Formally, define random variables $Y_p$, ${Y_1^0}_p$ and ${Y_1^1}_p$ as in eqn~\eqref{individualreports}. 
Given a dataset $X_b$ for $b \in \zo$ we generate a sample from $\threedistb{\eps_0}{p^*}$ as follows. Client number one (holding the first element of the dataset) reports a sample from ${Y_1^b}_{p^*}$. Clients $2,\ldots,n$ each report an independent sample from $Y_{p^*}$. We then count the total number of 0s and 1s. Note that a vector containing a permutation of the users responses contains no more information than simply the number of 0s and 1s, so we can consider these two representations as equivalent. This is why in Algorithm~\ref{postprocessing}, we can immediately turn the sample $(n_0,n_1)$ into a vector $y$ of 0s, 1s and 2s.

The mixture coefficients of the random variables $\Aldp[i]$ do not necessarily match those of ${Y_1^0}_{p^*}, {Y_1^1}_{p^*}$ and $Y_{p^*}$. However, for any $p<p^*$ we can define a post-processing function   $g(\cdot,p)$ such that $g({Y_1^0}_{p^*}, p)= {Y_1^0}_{p}, g({Y_1^1}_{p^*}, p)={Y_1^0}_{p}$ and $g(Y_{p^*}, p)=Y_{p}$. This function is given by $g(0,p)=0$ with probability $p/p^*$ and 2 otherwise, $g(1,p)=1$ with probability $p/p^*$ and 2 otherwise, $g(2,p)=2.$

Let $y\in\zot^n$ be a permutation of the local reports given by $Y_{1,p^*}^b$ and $n-1$ copies of $Y_{p^*}$; recall that this is equivalent to a sample from $\threedistb{\eps_0}{p^*}$. Given the hidden permutation $\pi$, we can generate a sample from $\shuffler(X_b)$ by sequentially transforming $y_t'=g(y_t, \decompprob{\Aldp[t](z_{1:t-1})})$ to obtain the correct mixture components, then sampling from the corresponding mixture component. The difficulty then lies in the fact that conditioned on a particular instantiation $y=v$, the permutation $\pi|_{y=v}$ is not independent of $b$.

The first thing to note is that if $v_t=0$ or $1$, then the corresponding mixture component ${\mathcal{Q}_1^0}^{(t)}(z_{1:t-1})$ or ${\mathcal{Q}_1^1}^{(t)}(z_{1:t-1})$, is independent of $\pi$. Therefore, in order to do the appropriate post-processing, it suffices to know the permutation $\pi$ restricted to the set of users who sampled $2$, $K=\pi(\{i: y_i = 2\})$. The set $K$ of users who select 2 is independent of $b$ since $Y_1^0$ and $Y_1^1$ have the same probability of sampling $2$. The probability of being included in $K$ is identical for each $i\in[2,\cdots,n]$, and slightly smaller for the first user. Given $n_K=|K|$, the probability of user 1 being included in $K$ is given by $\alpha = \mathbb{P}(Y_{1,p^*}^b=2\;|\; |K|=n_K) $. A closed form formula for $\alpha$ is given in Algorithm~\ref{postprocessing},
\begin{align*}
\alpha &= \mathbb{P}(Y_{1,p^*}^b=2\;|\; |K|=n_K)\\
&=\frac{\binom{n-1}{n_K-1}(1-e^{\eps_0}p-p)(1-2p)^{n_K-1}(2p)^{n-n_K}}{\binom{n-1}{n_k-1}(1-e^{\eps_0}p-p)(1-2p)^{n_K-1}(2p)^{n-n_K}+\binom{n-1}{n_K}(e^{\eps_0}p+p)(1-2p)^{n_K}(2p)^{n-1-n_K}}\\
    &= \frac{2p n_K(1-e^{\eps_0}p-p)}{2pn_K(1-e^{\eps_0}p-p)+(e^{\eps_0}p+p)(1-2p)(n-n_K)}
\end{align*}
where the second equality follows since there are $\binom{n-1}{n_K-1}$ possible choices for $K$ that include 1, and each has probability  $(1-e^{\eps_0}p-p)(1-2p)^{n_K-1}(2p)^{n-n_K}$, and there are $\binom{n-1}{n_K}$ choices for $K$ that do not include 1, and each has probability $(e^{\eps_0}p+p)(1-2p)^{n_K}(2p)^{n-1-n_K}$.
Any ordering of the users in $K$ is equally likely, so we can choose a random assignment. Now that we have an assignment for $K$, we can sample from the correct mixture components as desired.
\end{proof}

Theorem~\ref{mainanalyticalthm} is now a direct corollary of Lemma~\ref{shuffletobinoms} and the data processing inequality since Corollary~\ref{decomposition} implies we can always take $p^*=1/(e^{\eps_0}+1)$.

Lemma~\ref{shuffletobinoms} actually indicates that we can obtain better bounds for specific classes of local randomizers. An example of a family of local randomizers with a smaller $p^*$ is $k$- randomized response. In Appendix~\ref{kRR}, we show that we can obtain tight privacy amplification by shuffling bounds for $\kRR$ directly from Lemma~\ref{shuffletobinoms}.

\section{Asymptotically Optimal Bound for R\'enyi DP}
\label{sec:rdp}

Our bound relies on a general way to convert bounds on approximate DP for a range of $\delta$'s to a bound on RDP parameters for a certain range of moments $\alpha$. We use this general conversion together with the bounds on privacy amplification by shuffling given in Theorem \ref{thm:shuffling4adp} to derive our asymptotically optimal bound for RDP. We note that an asymptotically optimal bound is also implied by our conversion applied to the approximate DP bounds on privacy amplification by shuffling in \citep{FeldmanMT:2020}.

We start by stating our general conversion from approximate DP bounds to RDP. Previously such a conversion was proved only for pure differential privacy. Specifically, it is known that $\eps$-DP implies $(\alpha \eps^2/2, \alpha)$-RDP \citep{mironov2017renyi,Bun:2016}. We prove the claim under a condition on the tail of the privacy loss random variable that is essentially equivalent to approximate DP but is easier to work with.

\begin{theorem}
\label{thm:adp2rdp}
Let $P$ and $Q$ be probability distributions over the same domain $X$ such that for some $\sigma>0, \delta_{\min} \geq 0$ and $\eps_0>0$, for all $\delta > \delta_{\min}$, we have that:
\[\Pr_{x\sim Q}\left[  \left| \ln \left(\frac{P(x)}{Q(x)}\right)  \right| \geq \sqrt{\ln(1/\delta)} \cdot \sigma  \right] \leq 2\delta \ \]
and for all $x$, $\left|\ln \left(\frac{P(x)}{Q(x)}\right) \right| \leq \eps_0$.  Then, for all $\alpha \geq 1$,
\[ \E_{x\sim Q}\left[ \left(\frac{P(x)}{Q(x)}\right)^\alpha \right] \leq  e^{2\alpha^2\sigma^2} + 4 \delta_{\min} \cdot e^{\alpha \eps_0} .\]
In particular, if $\delta_{\min} =0$  then, for all $\alpha > 1$, $D^\alpha(P \| Q) \leq \frac{2\alpha^2 \sigma^2}{\alpha - 1} $. Further,
if $\delta_{\min} \leq e^{-\alpha \eps_0} \cdot \alpha^2\sigma^2 /4$ then, for all $\alpha > 1$,
\[D^\alpha(P \| Q) \leq \frac{3\alpha^2 \sigma^2}{\alpha - 1} .\]
\end{theorem}
Our proof relies on well-known properties of sub-Gaussian random variables. Specifically, we use the following lemma from \citep{rigollet-course}\footnote{Lemma~1.5 in \citep{rigollet-course} appears to implicitly assume that $Z$ is zero-mean as it omits the first order term $\alpha \E[Z]$ in the statement.}.
\begin{lemma}[\citep{rigollet-course}(Chap.~2, Lemmas~1.4 and 1.5)]
\label{lem:subg}
Let $Z$ be a random variable such that for every $t>0$,
\[\Pr [| Z| > t] \leq 2 \exp\left(-\frac{t^2}{2\sigma^2}\right) .\]
Then, for any $\alpha > 0$, it holds
\[\E[\exp(\alpha Z)] \leq \alpha \E[Z] + e^{4\alpha^2\sigma^2} .\]
\end{lemma}
\begin{proof}[Proof of Theorem~\ref{thm:adp2rdp}]
Let $Z(x)$ denote $\ln \left(\frac{P(x)}{Q(x)}\right)$. Then, by the assumptions, for all $0<t<\sigma \sqrt{\ln(1/\delta_{\min})}$, we have that:
\[\Pr_{x\sim Q}\left[  \left| Z(x) \right| \geq t \right] \leq 2 \exp\left(-\frac{t^2}{\sigma^2}\right) \ .\]
In addition, for all $x \in X$, $|Z(x)| \leq \eps_0$.

For $\eps_{\max} = \sigma \sqrt{\ln(1/\delta_{\min})}$, we denote by $Z'(x)$, $Z(x)$ truncated to the interval $[-\eps_{\max},\eps_{\max}]$, that is  $Z'(x) = \max\{-\eps_{\max},\min\{\eps_{\max},Z(x)\}\}$.
We note that now, for all $t>0$, we have that:
\[\Pr_{x\sim Q}\left[  \left|Z'(x) \right| \geq t \right] \leq 2 \exp\left(-\frac{t^2}{\sigma^2}\right) \ .\]
By Lemma~\ref{lem:subg} we have that, for all $\alpha > 0$,
\equ{
\E_{x\sim Q}[\exp(\alpha Z'(x))] \leq \alpha \E_{x\sim Q}[Z'(x)] + e^{2\alpha^2\sigma^2} .
}

This allows us to bound $\E_{x\sim Q}\left[ \left(\frac{P(x)}{Q(x)}\right)^\alpha \right]$ as follows:

\alequn{\E_{x\sim Q}\left[ \left(\frac{P(x)}{Q(x)}\right)^\alpha \right] &= \E_{x\sim Q}\left[ \exp(\alpha Z(x)) \right] \\
& \leq  \E_{x\sim Q}\left[  \exp(\alpha Z'(x)) \right] +   \E_{x\sim Q}\left[ \left| \exp(\alpha Z'(x)) - \exp(\alpha Z(x))  \right| \right]   \\
& \leq  \E_{x\sim Q}\left[  \exp(\alpha Z'(x)) \right] + 2 \delta_{\min} \cdot e^{\alpha \eps_0} \\
& \leq  \alpha \E_{x\sim Q}[Z'(x)] + e^{2\alpha^2\sigma^2} + 2 \delta_{\min} \cdot e^{\alpha \eps_0} \\
& \leq  \alpha \E_{x\sim Q}[Z(x)] + \alpha\E_{x\sim Q}[|Z'(x)-Z(x)|]  +  e^{2\alpha^2\sigma^2} + 2 \delta_{\min} \cdot e^{\alpha \eps_0} \\
& \leq  \alpha \E_{x\sim Q}[Z(x)] + 2 \delta_{\min} \alpha\eps_0 + e^{2\alpha^2\sigma^2} + 2 \delta_{\min} \cdot e^{\alpha \eps_0} \\
& \leq  \alpha \E_{x\sim Q}\left[\ln \left(\frac{P(x)}{Q(x)}\right)\right] + e^{2\alpha^2\sigma^2} + 4 \delta_{\min} \cdot e^{\alpha \eps_0} \\
& =  -\alpha \mbox{KL}(Q\|P) + e^{2\alpha^2\sigma^2} + 4 \delta_{\min} \cdot e^{\alpha \eps_0} \\
& \leq  e^{2\alpha^2\sigma^2} + 4 \delta_{\min} \cdot e^{\alpha \eps_0} ,
}
where we used the fact that $\mbox{KL}(Q\|P)$ is always non-negative.
This shows the first part of the claim.
Now if $\delta_{\min} \leq e^{-\alpha \eps_0} \cdot \alpha^2\sigma^2 /4$ then
\[\E_{x\sim Q}\left[ \left(\frac{P(x)}{Q(x)}\right)^\alpha \right] \leq e^{2\alpha^2\sigma^2} + 4 \delta_{\min} \cdot e^{\alpha \eps_0} \leq e^{2\alpha^2\sigma^2} + \alpha^2\sigma^2 \leq e^{3\alpha^2\sigma^2} ,\] where we used that for $a\geq 0$ and any $b$, $e^a + b = e^a (1+ b/e^a) \leq e^a (1+ b) \leq e^a e^b = e^{a+b}$. 
By definition of $D^\alpha$, we now have that $D^\alpha(P \| Q) \leq \frac{3\alpha^2 \sigma^2}{\alpha - 1} $.
\end{proof}
We note that, for $\alpha \geq 2$, $\frac{\alpha}{\alpha - 1} \leq 2$. Thus, for $\alpha \geq 2$, $D^\alpha(P \| Q) \leq 6 \alpha \sigma^2$ and for $\alpha \in [1,2]$ we can simply upper-bound $D^\alpha(P \| Q) \leq D^2(P \| Q) \leq 12 \sigma^2$. 

The tail bound in the condition of Theorem~\ref{thm:adp2rdp}  is somewhat stronger than what is implied by $( \sqrt{\ln(1/\delta)} \cdot \sigma , 2\delta)$-DP. However, $D_{e^\eps}(P\|Q) \leq \delta$ and $D_{e^\eps}(Q\|P) \leq \delta$ imply that $\Pr[|\ln(\frac{P}{Q} )| \geq  2\eps] \leq 2 \delta/ (e^{2\eps} - e^\eps) \leq 2\delta/\epsilon$ \cite[Lemma~9]{CanonneKS20}. Thus the conversion can be applied to  $( \sqrt{\ln(1/\delta)} \cdot \sigma , 2\delta)$-DP with small adjustements in the final bounds. For our application we bypass the need to convert from the hockey-stick divergence to the tail bound by directly proving  a bound on the tail of privacy loss random variable in Lemma~\ref{lem:tail-bound} (which also implies Theorem~\ref{thm:shuffling4adp}). This gives us the following corollary (the proof can be found in Appendix~\ref{app:rdp-cor}).
\begin{corollary}\label{RDPbound} \textcolor{magenta}{Errata: the proof of this theorem was affected by the error in the proof of Lemma 3.5. However, the theorem still holds with a slightly different constant by using results from~\cite{Feldman:2022focs}. See Corollary~\ref{RDPbound-corrected} for the updated constants.}
 For any domain $\mathcal{D}$, let $\Aldp[i]:\out[1]\times\cdots\times\out[i-1]\times\mathcal{D}\to\out[i]$ for $i\in[n]$ (where $\out[i]$ is the range space of $\Aldp[i]$) be a sequence of algorithms such that $\Aldp[i](z_{1:i-1}, \cdot)$ is an $\eps_0$-DP local randomizer for all values of auxiliary inputs $z_{1:i-1}\in\out[1]\times\cdots\times\out[i-1]$.
 Let $\shuffler:\mathcal{D}^n\to\out[1]\times\cdots\times \out[n]$ be the algorithm that given a dataset $x_{1:n}\in\mathcal{D}^n$, samples a uniform random permutation $\pi$ over $[n]$, then sequentially computes $z_i=\Aldp[i](z_{1:i-1}, x_{\pi(i)})$ for $i\in[n]$ and outputs $z_{1:n}$.
 Then there exists a constant $c$ such that for any $\alpha < \frac{n}{16 \eps_0 e^\eps_0}$, $\shuffler$ is $(\alpha \rho,\alpha)$-RDP,  where
\begin{equation}\label{rhobound}
\rho \le c \cdot (1-e^{-\eps_0})^2\frac{e^{\eps_0}}{n} .
\end{equation}
In particular, for $\eps_0 \geq 1$,$\rho \le  \frac{c e^{\eps_0}}{n}$.
\end{corollary}

\section{Improved Bounds for Specific Randomizers}\label{kRR}

In Section~\ref{generalamplificationreduction} we focused on a general amplification by shuffling bound that holds for all sets of local randomizers. Lemma~\ref{shuffletobinoms} actually indicates that we can obtain better bounds for specific classes of local randomizers. Namely, when $p$ as defined in Lemma~\ref{shuffletobinoms} is less than $1/(e^{\eps_0}+1)$.
The following lemma makes explicit the fact that the smaller $p^*$ is, the more amplification we can obtain. Given two random variables $X$ and $Y$, we will use the notation $X\disteq Y$ to denote that $X$ and $Y$ have the same distribution. That is if $X$ is a random variable over a finite space $\mathcal{S}$ and $Y$ is a random variable over a finite space $\mathcal{S}'$ then $X\disteq Y$ if there exists a invertible mapping $g$ such that for all $s\in\mathcal{S}$, $\Pr(X=s)=\Pr(Y=g(s))$.

\begin{lemma}\label{maxdivergence}
For any $p,p'\in[0,1]$ and $\eps>0$, if $p<p'$ then \[\dalpha{e^{\eps}}(\threedistone{\eps_0}{p}\|\threedisttwo{\eps_0}{p})\le \dalpha{e^{\eps}}(\threedistone{\eps_0}{p'}\|\threedisttwo{\eps_0}{p'})\]
\end{lemma}

An example of a family of local randomizers with a smaller $p^*$ is $k$- randomized response. In fact, we can show that Lemma~\ref{shuffletobinoms} gives a tight upper bound on the privacy amplification by shuffling for this family of local randomizers,. For any $k\in\mathbb{N}$ and $\eps_0 >0$, the $k$-randomized response $\kRR\colon [k]\to[k]$ is defined as \[\kRR(x) = \begin{cases} x & \text{with probability } \frac{e^{\eps_0}-1}{e^{\eps_0}+k-1}\\  y\sim \unif{[k]} & \text{with probability } \frac{k}{e^{\eps_0}+k-1}\end{cases},\] where $\unif{[k]}$ is the uniform distribution over $[k]$. That is, with probability $\frac{e^{\eps_0}-1}{e^{\eps_0}+k-1}$ the true data point is reported, and otherwise a random value is reported.

\begin{theorem}[$\kRR$ upper bound]\label{kRRthm}\textcolor{magenta}{Errata: this theorem is affected by the error in the proof of Lemma 3.5 but a slight variant holds and is given in Theorem~\ref{kRRthmnew}.} For a domain $\mathcal{D}$, let $\Aldp[i]:\out[1]\times\cdots\times\out[i-1]\times\mathcal{D}\to\out[i]$ for $i\in[n]$ (where $\out[i]$ is the range space of $\Aldp[i]$) be a sequence of algorithms such that $\Aldp[i](z_{1:i-1}, \cdot)$ is a $\eps_0$-DP local randomizer for all values of auxiliary inputs $z_{1,i-1}\in\out[1]\times\cdots\times\out[i-1]$. Let $\shuffler:\mathcal{D}^n\to\out[1]\times\cdots\times \out[n]$ be the algorithm that given a dataset $x_{1:n}\in\mathcal{D}^n$, samples a uniformly random permutation $\pi$, then sequentially computes $z_i=\Aldp[i](z_{1:i-1}, x_{\pi(i)})$ for $i\in[n]$ and outputs $z_{1:n}$. Assume that for some $k\in\mathbb{N}$ we have that for all $i\in[n]$, there exists a function $f^{(i)}:\out[1]\times\cdots\times\out[i-1]\times\mathcal{D}\to[k]$ such that $\Aldp[i](z_{1:i-1}, x)=\kRR(f^{(i)}(z_{1:i-1},x))$. Then for any distance measure $D$ that satisfies the data processing inequality, and any two neighbouring datasets $X_0$ and $X_1$, \[D(\shuffler(X_0), \shuffler(X_1))\le D\left(\threedistone{\eps_0}{\frac{1}{e^{\eps_0}+k-1}}, \threedisttwo{\eps_0}{\frac{1}{e^{\eps_0}+k-1}}\right)\]
\end{theorem}
\begin{proof}
Let $X_0$ and $X_1$ be neighbouring datasets. If for all $i\in[n]$, $\Aldp[i](z_{1:i-1}, x)=\kRR(f^{(i)}(z_{1:i-1},x))$, then for all $i$, the probability density function of $\Aldp[i]$ only takes on two values $e^{\eps_0}p$ and $p$ where $p=1/(e^{\eps_0}+k-1)$. Thus, as in Corollary~\ref{decomposition}, this allows us to show that $\decompprob{\Aldp[i](z_{1:i-1})}\le 1/(e^{\eps0}+k-1)$ and hence by Lemma~\ref{shuffletobinoms} and the data processing inequality, $D(X_0,X_1)\le D(\threedistone{\eps_0}{p}, \threedisttwo{\eps_0}{p})$ as required.
\end{proof}

\begin{theorem}[$\kRR$ lower bound]\label{kRRlowerbound}
Given $k\in\mathbb{N}$, $k \geq 3$ let the data domain be $\mathcal{D}=[k]$. Let $\shuffler:[k]^n\to[k]$ be the algorithm that given a dataset $x_{1:n}\in[k]^n$, samples a uniformly random permutation $\pi$, then computes $z_i=\kRR(x_{\pi(i)})$ for $i\in[n]$ and outputs $z_{1:n}$. Then there exists neighbouring datasets $X_0$ and $X_1$ such that for any distance measure $D$ that satisfies the data processing inequality, \[D\left(\threedistone{\eps_0}{\frac{1}{e^{\eps_0}+k-1}},\threedisttwo{\eps_0}{\frac{1}{e^{\eps_0}+k-1}}\right)\le D(\shuffler(X_0), \shuffler(X_1))\]
\end{theorem}
\begin{proof}
Let $X_0$ be the data set where the first individual has the value 1, and the $n-1$ remaining individuals have value 3. Let $X_1$ be the data set where the first individual has the value 2, and the $n-1$ remaining individuals have value 3. Consider the function $g:[k]\to [3]$ that maps 1 and 2 to themselves, then $g(i)=3$ for all $i\in\{3,\cdots,k\}$. Let $g^n:[k]^n\to\mathbb{N}^2$ be $g(z_{1:n})=(\sum_{i=1}^n \chi(z_i=1), \sum_{i=1}^n \chi(z_i=2))$, where $\chi(y)=1$ if the logical condition $y$ is true, and 0 otherwise. Then $g(\shuffler(X_0))= \threedistone{\eps_0}{1/(e^{\eps_0}+k-1)}$ and $g(\shuffler(X_1))= \threedisttwo{\eps_0}{1/(e^{\eps_0}+k-1)}$. The result then follows from the data processing inequality.
\end{proof}

\section{Numerical Results}

In this section, we provide a numerical evaluation of our bound.

\subsection{Approximate Differential Privacy}

In Figure~\ref{approxgraphs}, \texttt{This work} is the bound from Theorem~\ref{mainanalyticalthm} with the hockey-stick divergence computed numerically. We compare to the best numerical amplification bounds from prior work by ~\citet{FeldmanMT:2020} (\texttt{FMT'20}) \footnote{\texttt{FMT'20} was produced using code released by Feldman, McMillan and Talwar at \url{https://github.com/apple/ml-shuffling-amplification}}. We also compare to the lower bound produced by computing the privacy amplification for the specific local randomizers $\2RR$ (\texttt{2RR, lower bound}) and $\3RR$ (\texttt{3RR, lower bound}). 
Figure~\ref{ADPratio} compares the bound presented in this work to that of \citet{FeldmanMT:2020}. As expected, our new bound is tighter than those in \citep{FeldmanMT:2020} in every regime. The deviation between the two bounds is largest when $\eps_0$ is large. As $\eps_0$ increases, the upper bound presented in this work approaches the lower bound obtained by directly computing the privacy amplification bounds for the specific mechanisms $\2RR$ and $\3RR$. The peak of the "This work" graph, where the upper bound from this work is furthest from the lower bound occurs exactly at the point when the $\3RR$ lower bound starts dominating the $\2RR$ lower bound. That is, to the left of the peak, $\3RR$ amplifies better than $\2RR$, and to the right of the peak, $\2RR$ amplifies better than $\3RR$.

\begin{figure*}
  \centering
  \begin{subfigure}{0.49\textwidth}
    \includegraphics[width=1\textwidth]{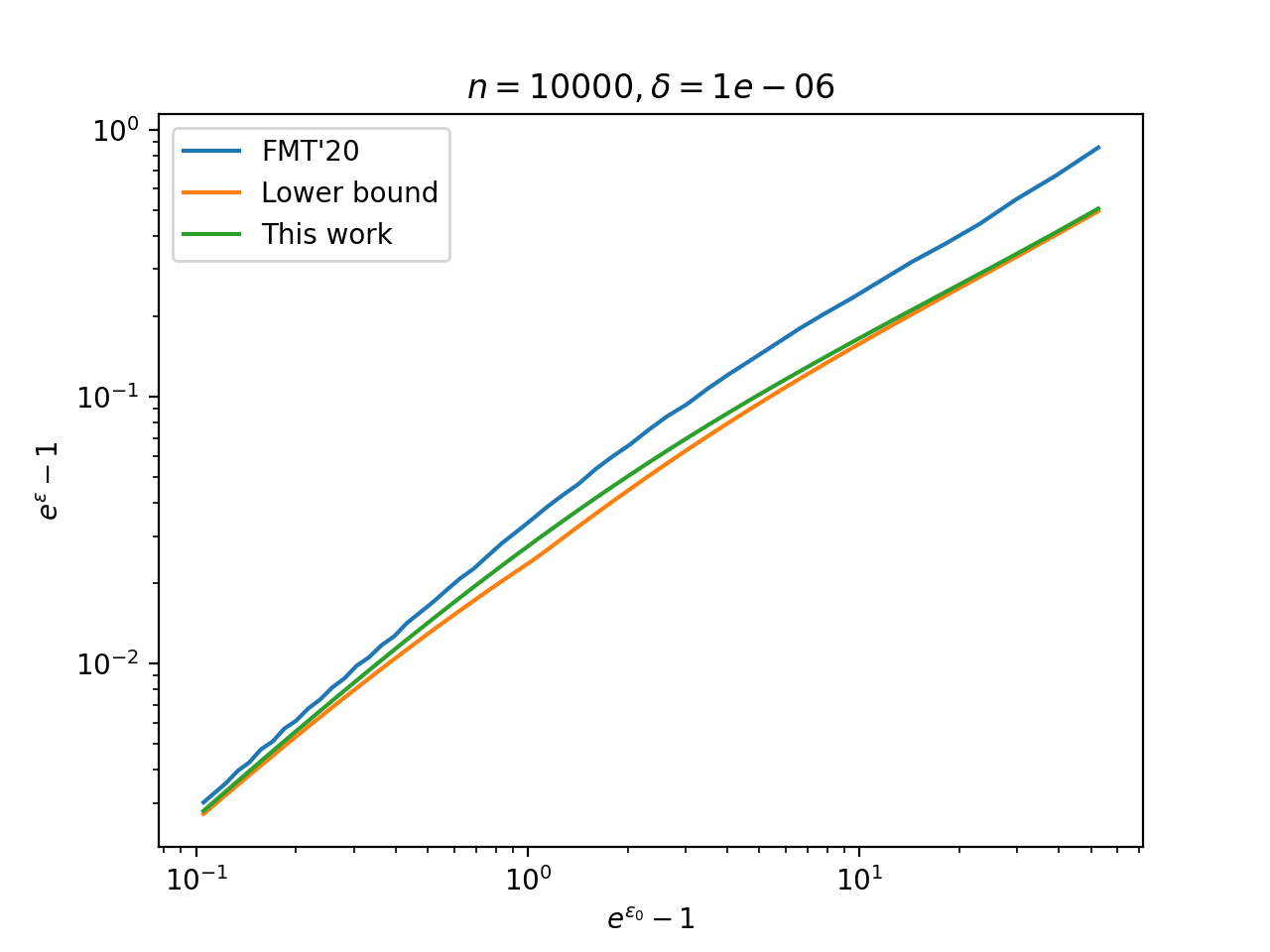}
  \caption{Comparison of privacy guarantees.}\label{approxgraphs}
  \end{subfigure}
  \begin{subfigure}{0.49\textwidth}
  \includegraphics[width=1\textwidth]{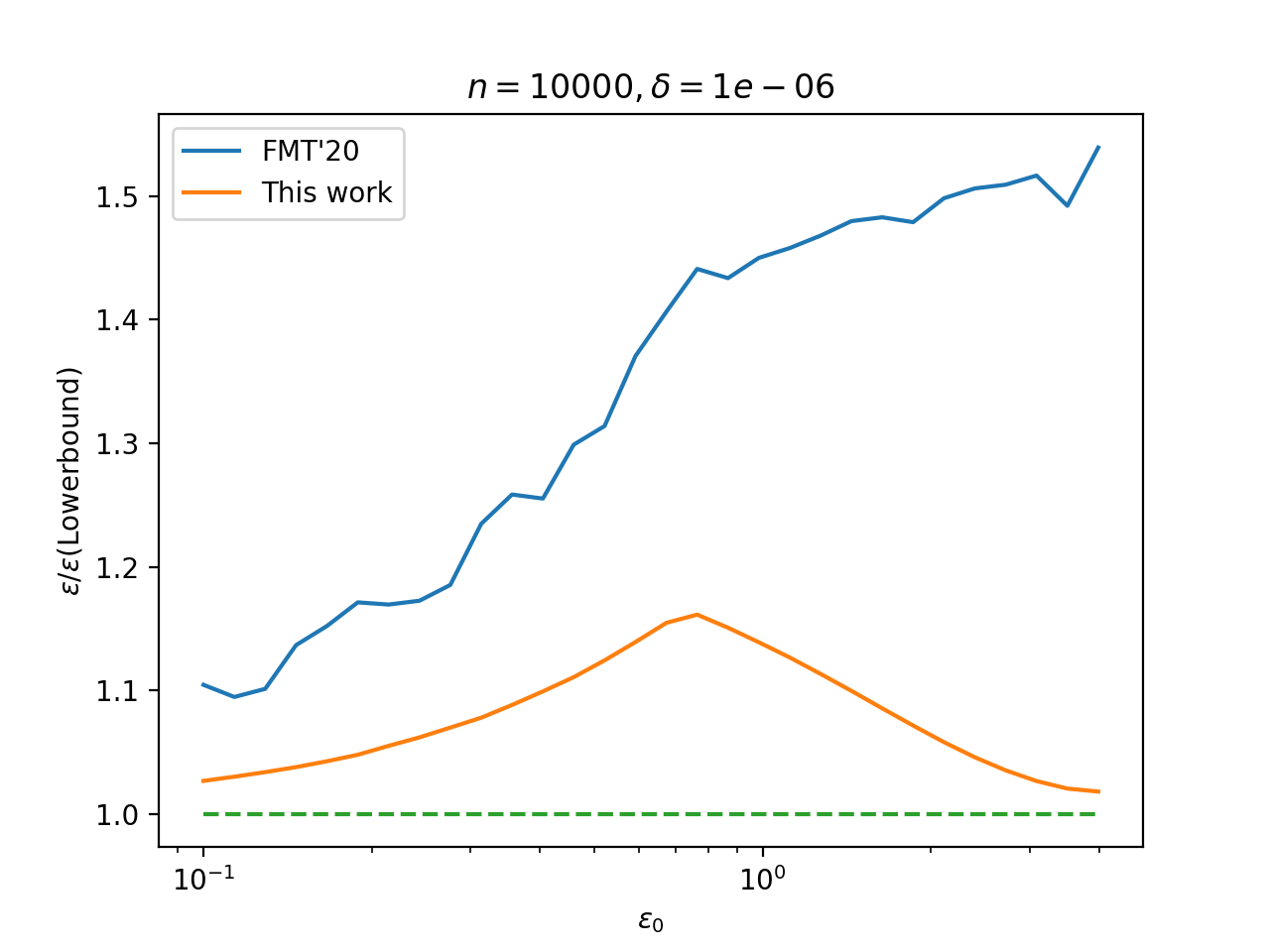}
    \caption{Comparison of privacy guarantees relative to lower bound.}
    \label{ADPratio}
  \end{subfigure}
  \caption{Comparison of the approximate DP shuffling amplification bound presented in this work, and that of \cite{FeldmanMT:2020}, the best bound from prior work. The lower bound on privacy amplification is obtained from $\epsilon(\rm Lower bound) = \max\{\epsilon(\2RR), \epsilon(\3RR)\}$ where $\epsilon(\2RR)$ and $\epsilon(\3RR)$ are the lower bounds on the privacy amplification of the specific mechanisms $\2RR$ and $\3RR$.}
\end{figure*}

Our experiments running on a 2021 MacBook Pro took 30 minutes for $n = 1,000,000$ (in SM), and 20 minutes for $n = 1000$. The bulk of this time is spent on computing our lower bounds. The upper bound computation for a fixed $\eps_0$ and $n=10^6$ takes about 1 minute. The lower bound for $\2RR$ similarly runs is about a minute for $n=10^6$. The lower bound for $\3RR$ runs in about 3 minutes for $n=10^4$, we did not run this algorithm for larger values of $n$.

\subsection{R\'enyi Differential Privacy}

In Figure~\ref{varyingalphafig} we show the privacy amplification bound for R\'enyi differential privacy as a function of $\alpha$. As expected, our bound always improves over \cite{FeldmanMT:2020}, and is very close to the lower bound for some settings of $\alpha$.

In Figure~\ref{varyingTfig}, we plot the privacy guarantee for $T$ adaptively composed outputs of a shuffler with each shuffler operating on $n$, $\epsilon_0$-DP local randomizers.
The advanced composition theorem quantifies the privacy guarantee after composing $T$ $(\epsilon,\delta)$-DP algorithms. However, we can obtain tighter privacy guarantees by computing the composition guarantees in terms of RDP, then converting back to approximate DP [ACGMMTZ16]. In Figure~\ref{varyingTfig}, we compare two methods for computing the resulting privacy guarantee. \texttt{This work, via Approximate DP} computes the amplification in terms of approximate DP then uses advanced composition [KOV15, Theorem 4.3]. \texttt{This work, via RDP} computes the amplification in terms of RDP, using composition in terms of RDP [Mir17], then converts to Approximate DP [CKS20, Proposition 12]. We also compare to the R\'enyi composition version of the best known prior bounds \citep{FeldmanMT:2020}.

\begin{figure}
    \centering
    \includegraphics[scale=0.7]{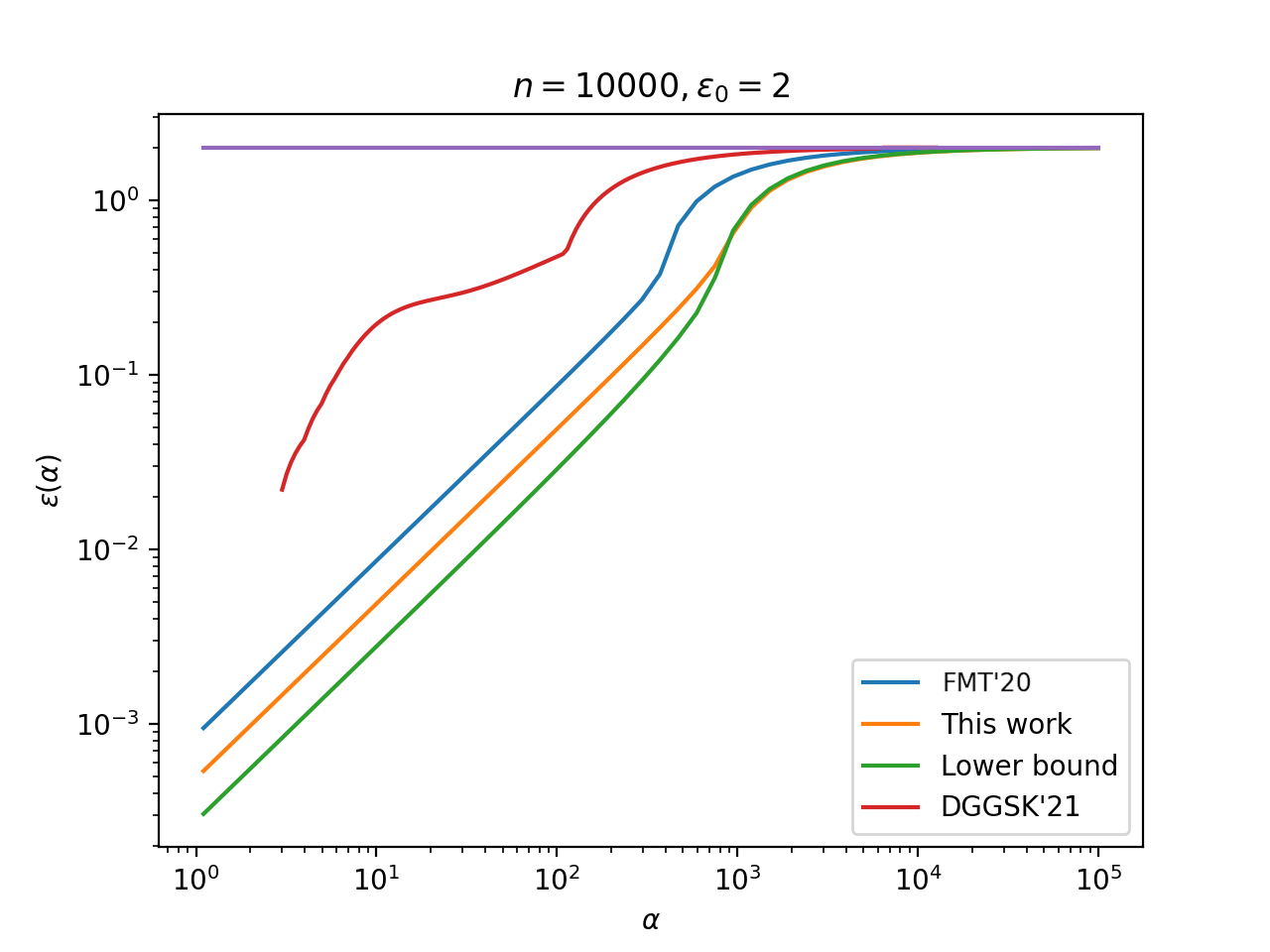}
    \caption{Comparison of the bounds on $\epsilon(\alpha)$ for RDP for fixed $\eps_0=4$ and $\alpha$ ranging between 1.1 and $10^5$. The horizontal line is at $\eps_0$.}
    \label{varyingalphafig}
\end{figure}

\begin{figure}
    \centering
    \includegraphics[scale=0.7]{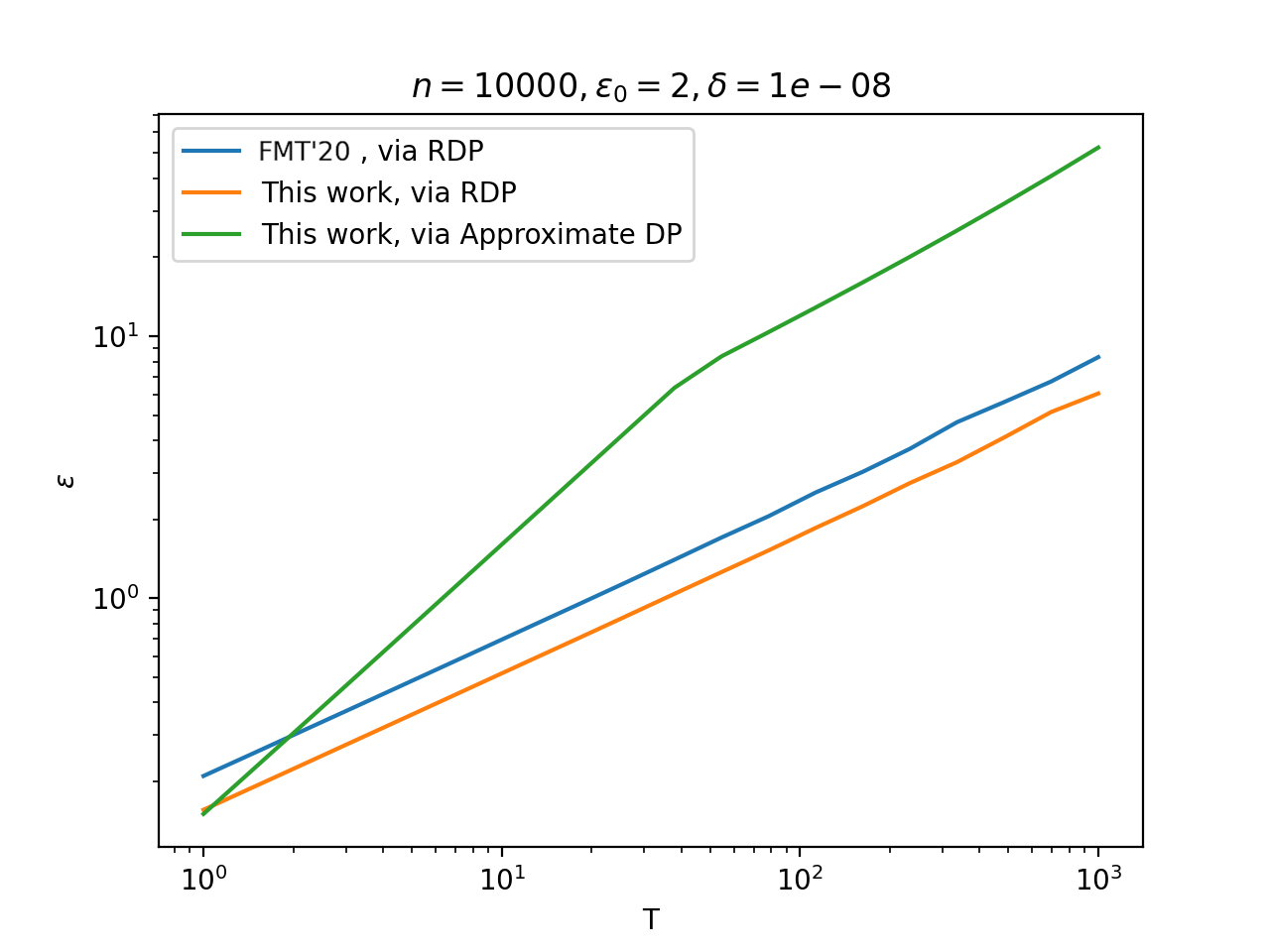}
    \caption{Approximate DP guarantees for composition of a sequence of $T$ results of amplification by shuffling.}
    \label{varyingTfig}
\end{figure}

\subsection{$k$-randomised response}

In Section~\ref{kRR} we showed that the analysis presented in this paper can obtain tight privacy amplification by shuffling bounds for the specific randomizers $\kRR$. In Figure~\ref{kRRfig} we can see that the improved bounds for $\kRR$ are indeed tighter than the general bound, and the privacy guarantee on the output of the shuffler improves as $k$ increases, as expected. We also compare to the best known bound from prior work \cite{Balle:2019}.

\begin{figure}
    \centering
    \includegraphics[scale=0.7]{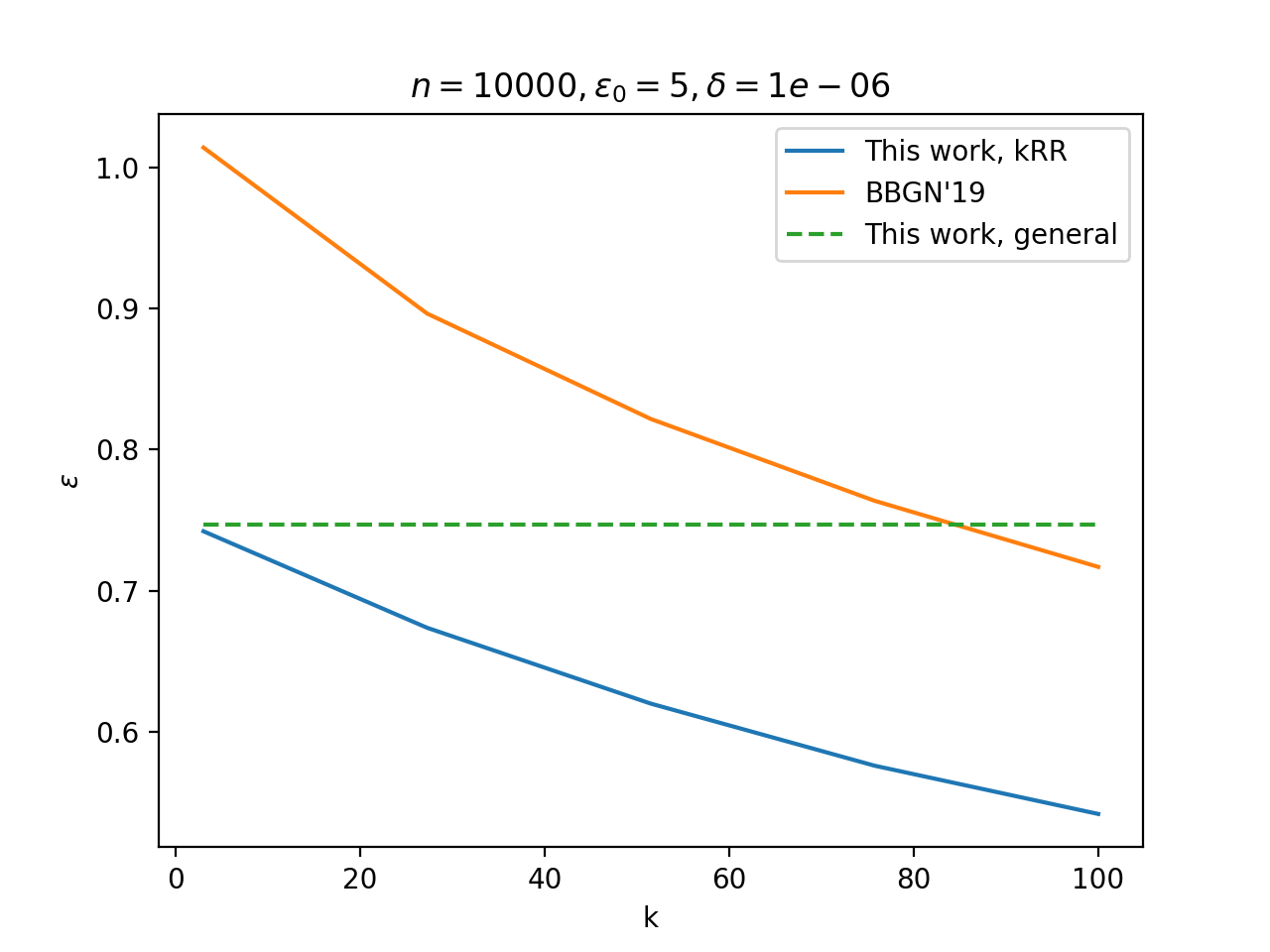}
    \caption{Privacy amplification by shuffling for $\kRR$ as a function of $k$ using bound from Theorem~\ref{kRRthm}. The horizontal dashed line corresponds to the general upper bound from Theorem~\ref{mainanalyticalthm}. }
    \label{kRRfig}
\end{figure}

\section{Errata}
\label{errata}
\subsection{The Error}

Let the set $K$ be as defined in the proof of Lemma 3.5. The problem comes from deciding whether $1\in K$. The proof had assumed that the probability that $1\in K$ only depended on $n_2$ and was the same regardless of whether the input was $X_0$ or $X_1$ since $\Pr(Y_{1,p}^0=2)=\Pr(Y_{1,p}^1=2)$. However, the probability that $1\in K$ actually also depends on $n_0$ and $n_1$, whose distribution depends on whether the input database was $X_0$ or $X_1$. For example, if $n_0>n_1$, then the probability $1\in K$ is higher if the input is $X_1$ than $X_0$.

This error does not arise if the local randomizers used always satisfy a decomposition with $p=\frac{1}{e^{\eps_0}+1}$, since in this setting $Y_{1,p}^0$ and $Y_{1,p}^1$ never output $2$, and hence this issue never arises.

\subsection{A General Statement for a Restricted Class of Local Randomizers}

Let $\extremalset{\eps_0}$ be the set of all local randomizers that satisfy the decomposition in Corollary 3.4 with $p=\frac{1}{e^{\eps_0}+1}$. That is, a local randomizer $\lr\colon \D\to \cS$ is in $\extremalset{\eps_0}$ if and only if it is an $\eps_0$-DP local randomizer and for any $n+1$ inputs $x_1^0, x_1^1, x_2, \ldots, x_n \in\mathcal{D}$, there lexists distributions $\mathcal{Q}_1^0$, $\mathcal{Q}_1^1$, $\mathcal{Q}_2, \ldots, \mathcal{Q}_n$ such that
\begin{align}
\nonumber\mathcal{R}(x_1^0) &= \frac{e^{\eps_0}}{e^{\eps_0}+1}\mathcal{Q}_1^0+\frac{1}{e^{\eps_0}+1}\mathcal{Q}_1^1,\\
\nonumber\mathcal{R}(x_1^1) &= \frac{1}{e^{\eps_0}+1}\mathcal{Q}_1^0+\frac{e^{\eps_0}}{e^{\eps_0}+1}\mathcal{Q}_1^1\\
\forall i\in[2,n], \;\;\mathcal{R}(x_i) &= \frac{1}{e^{\eps_0}+1}\mathcal{Q}_1^0+\frac{1}{e^{\eps_0}+1}\mathcal{Q}_1^1+\frac{e^{\eps_0}-1}{e^{\eps_0}+1}\mathcal{Q}_i. \label{eq:decomposition}
\end{align}

Since Lemma 3.5 still holds for this set of local randomizers, the upper bounds in Theorem 3.1, Theorem 3.2 and Corollary 4.3 still hold for local randomizers in this set. 

\begin{theorem}[Replacement for Theorem 3.1]\label{newtheorem31}
For a domain $\mathcal{D}$, let $\Aldp[i]:\out[1]\times\cdots\times\out[i-1]\times\mathcal{D}\to\out[i]$ for $i\in[n]$ (where $\out[i]$ is the range space of $\Aldp[i]$) be a sequence of algorithms such that $\Aldp[i](z_{1:i-1}, \cdot)\in\extremalset{\eps_0}$ for all values of auxiliary inputs $z_{1:i-1}\in\out[1]\times\cdots\times\out[i-1]$. Let $\shuffler:\mathcal{D}^n\to \out[1]\times\cdots\times \out[n]$ be the algorithm that given a dataset $x_{1:n}\in\mathcal{D}^n$, samples a permutation $\pi$ uniformly at random, then sequentially computes $z_i=\Aldp[i](z_{1:i-1}, x_{\pi(i)})$ for $i\in[n]$ and outputs $z_{1:n}$. Let $X_0$ and $X_1$ be two arbitrary neighboring datasets in $\mathcal{D}^n$. Then for any distance measure $D$ that satisfies the data processing inequality,
\[D(\shuffler(X_0)\|\shuffler(X_1))\le D\left(\threedistonemain{\eps_0}\Big\|\threedisttwomain{\eps_0}\right).\]
\end{theorem}

Since Theorem 3.2 and Corollary 4.3 are bounds on $D\left(\threedistonemain{\eps_0}\Big\|\threedisttwomain{\eps_0}\right)$, the bounds provided in these results hold for the same set-up at Theorem~\ref{newtheorem31}.

\paragraph{What local randomizers belong in $\extremalset{\eps_0}$?} The set $\extremalset{\eps_0}$ does not contain all $\eps_0$-DP local randomizers, although it does contain many important local randomizers. For example, for any $k\in\mathbb{N}$, $\texttt{kRR}\in\extremalset{\eps_0}$. 
Further, $\extremalset{\eps_0}$ is closed under convex combinations. 

RAPPOR~\cite{Erlingsson:2014} is also in $\extremalset{\eps_0}$: given $K\in\mathbb{N}$, $\alpha,\beta\in[0,1]$ and $x\in[K]$, the RAPPOR local randomizer first encodes $x$ as a vector of length $K$, which is 1 in the $x$th position, and 0 elsewhere. It then independently tosses a coin for each coordinate (with bias $\alpha$ if the coordinate value is 1 and bias $\beta$ if the coordinate value is 0) and replaces the coordinate value with the result of the coin flip. Whenever $\alpha\ge \max\{1/2, \beta\}$ and $\alpha$ and $\beta$ are chosen such that this local randomizer is $\eps_0$-DP, then it is also in $\extremalset{\eps_0}$ with $\mathcal{Q}_1^0, \mathcal{Q}_1^1$ and $\mathcal{Q}_2$ given by 
\begin{align*}
    \mathcal{Q}_1^0 &= \left[\Ber\left(\frac{\alpha e^{\eps_0}-\beta}{e^{\eps_0}-1}\right), \Ber\left(\frac{\beta e^{\eps_0}-\alpha}{e^{\eps_0}-1}\right), \Ber(\beta), \cdots, \Ber(\beta)\right]\\
    \mathcal{Q}_1^0 &= \left[\Ber\left(\frac{\beta e^{\eps_0}-\alpha}{e^{\eps_0}-1}\right), \Ber\left(\frac{\alpha e^{\eps_0}-\beta}{e^{\eps_0}-1}\right), \Ber(\beta), \cdots, \Ber(\beta)\right]\\
    \mathcal{Q}_2 &= \left[\Ber\left(\frac{\beta e^{\eps_0}-\alpha}{e^{\eps_0}-1}\right), \Ber\left(\frac{\beta e^{\eps_0}-\alpha}{e^{\eps_0}-1}\right), \Ber\left(\frac{(e^{\eps_0}+1)\alpha-2\beta}{e^{\eps_0}-1}\right), \Ber(\beta), \cdots, \Ber(\beta)\right]
\end{align*}
where $x_1^0=1, x_1^1=2$ and $x_2=3$.

\subsection{Asymptotically Optimal Bound for R\'enyi DP}
Here we restate Corollary~\ref{RDPbound} with the corrected constant (in the condition on $\alpha$)	. The proof now relies on the reduction in \cite{FeldmanMT:2020} (which is identical up to a factor of at most 2 to the one claimed in Theorem~\ref{mainanalyticalthm}). The proof is in Appendix~\ref{app:rdp-cor}.
\begin{corollary}\label{RDPbound-corrected}
		For any domain $\mathcal{D}$, let $\Aldp[i]:\out[1]\times\cdots\times\out[i-1]\times\mathcal{D}\to\out[i]$ for $i\in[n]$ (where $\out[i]$ is the range space of $\Aldp[i]$) be a sequence of algorithms such that $\Aldp[i](z_{1:i-1}, \cdot)$ is an $\eps_0$-DP local randomizer for all values of auxiliary inputs $z_{1:i-1}\in\out[1]\times\cdots\times\out[i-1]$.
		Let $\shuffler:\mathcal{D}^n\to\out[1]\times\cdots\times \out[n]$ be the algorithm that given a dataset $x_{1:n}\in\mathcal{D}^n$, samples a uniform random permutation $\pi$ over $[n]$, then sequentially computes $z_i=\Aldp[i](z_{1:i-1}, x_{\pi(i)})$ for $i\in[n]$ and outputs $z_{1:n}$.
		Then there exists a constant $c$ such that for any $\alpha < \frac{n}{32 \eps_0 e^\eps_0}$, $\shuffler$ is $(\alpha \rho,\alpha)$-RDP,  where
		\begin{equation}\label{rhobound}
			\rho \le c \cdot (1-e^{-\eps_0})^2\frac{e^{\eps_0}}{n} .
		\end{equation}
		In particular, for $\eps_0 \geq 1$,$\rho \le  \frac{c e^{\eps_0}}{n}$.
	\end{corollary}

\subsection{Improved Bounds for Specific Randomizers}

Our replacement Lemma 3.5 relies on a slightly different decomposition, and multinomial distributions that are over $\mathbb{N}^4$ rather than $\mathbb{N}^3$. This allows us to avoid the issue that arose in the proof of Lemma 3.5
Given an $\eps_0$-DP local randomizer $\mathcal{R}:\mathcal{D}\to\mathcal{S}$, and any $n+1$ inputs $x_1^0, x_1^1, x_2, \cdots, x_n \in\mathcal{D}$, suppose there exists $p\in[0,1/(e^{\eps_0}+1)]$ and $q\in[0, 1-2p]$ and distributions $\mathcal{Q}_1^0$, $\mathcal{Q}_1^1$, $\mathcal{Q}_1, \mathcal{Q}_2, \cdots, \mathcal{Q}_n$ such that
\begin{align}
\nonumber\mathcal{R}(x_1^0) &= e^{\eps}p\mathcal{Q}_1^0+p\mathcal{Q}_1^1+(1-p-e^{\eps_0}p)\mathcal{Q}_1,\\
\nonumber\mathcal{R}(x_1^1) &= p\mathcal{Q}_1^0+e^{\eps_0}p\mathcal{Q}_1^1+(1-p-e^{\eps_0}p)\mathcal{Q}_1\\
\forall i\in[2,n], \;\;\mathcal{R}(x_i) &= p\mathcal{Q}_1^0+p\mathcal{Q}_1^1+q\mathcal{Q}_1+(1-2p-q)\mathcal{Q}_i. \label{eq:decomposition2}
\end{align}
Corollary~\ref{decomposition} states that such a decomposition always exists with $q=0$. With a slight modification to that proof, we can show that such a decomposition always exists with $q\ge e^{-\eps}(1-p-e^{\eps}p)$. We begin by formally defining the distributions $\threedistqone{\eps_0}{p}{q}$ and $\threedistqtwo{\eps_0}{p}{q}$.
Define the random variable $Y_{p,q}$
\begin{equation}\label{individualreportsq}
Y_{p,q} = \begin{cases} 0 & \text{w.p.   } p\\ 1 & \text{w.p.   } p \\ 2 & \text{w.p. }q \\ 3 & \text{w.p.   } 1-2p-q \end{cases}
\end{equation}
For $b\in\{0,1\}$, to obtain a sample from $\threedistqb{\eps_0}{p}{q}$, sample one copy from $Y_{1,p}^b$ (as originally defined in eq.~\eqref{individualreports}) and $n-1$ copies of $Y_{p,q}$, then output $(n_0,n_1,n_2)$ where $n_0, n_1$ and $n_2$ are the total number of 0s, 1s and 2s, respectively.

\begin{lemma}[Replacement for Lemma 3.5]\label{shuffletobinoms2new}
For a domain $\mathcal{D}$, let $\Aldp[i]:\out[1]\times\cdots\times\out[i-1]\times\mathcal{D}\to\out[i]$ for $i\in[n]$ (where $\out[i]$ is the range space of $\Aldp[i]$, and $\out[i]$ is finite for all $i$) be a sequence of algorithms such that $\Aldp[i](z_{1:i-1}, \cdot)$ is an $\eps_0$-DP local randomizer for all values of auxiliary inputs $z_{1:i-1}\in\out[1]\times\cdots\times\out[i-1]$. Let $\shuffler:\mathcal{D}^n\to \out[1]\times\cdots\times \out[n]$ be the algorithm that given a dataset $x_{1:n}\in\mathcal{D}^n$, samples a permutation $\pi$ uniformly at random, then sequentially computes $z_i=\Aldp[i](z_{1:i-1}, x_{\pi(i)})$ for $i\in[n]$ and outputs $z_{1:n}$. Let $X_0$ and $X_1$ be two arbitrary neighboring datasets in $\D^n$. Let $(p^*, q^*)$ be such that for all $i\in[n]$ and $z_{1:i-1}\in\out[1]\times\cdots\times\out[i-1]$, $\Aldp[i](z_{1:i-1},\cdot)$ satisfies eqn~\eqref{eq:decomposition2} with $(p,q)$ such that $p<p^*$ and $q>q^*+2(p^*-p)$.
Then there exists a post-processing function $f$ such that $\shuffler(X_0)$ is distributed identically to $f(\threedistqone{\eps_0}{p^*}{q^*})$ and $\shuffler(X_1)$ is distributed identically to  $f(\threedistqtwo{\eps_0}{p^*}{q^*})$.
\end{lemma}

The proof of Lemma~\ref{shuffletobinoms2new} is very similar to the (flawed) proof of Lemma 3.5. The key difference is that we only need to assign identities to users who output ``3'', and user 1 never outputs ``3''.

\begin{algorithm}
\caption{Post-processing function, $f$} \label{postprocessing2}
  {\textbf{Input:} $x^0_1,x^1_1,x_2,\ldots, x_n$; $(n_0, n_1, n_2)\in\mathbb{N}^3$}\\
  {Sample $y \in \{0,1,2,3\}^n$ to be a random permutation of a vector with $n_0$ 0s, and $n_1$ 1s, $n_2$ 2s, and $n-n_0-n_1-n_2$ 3s.}\\
  {$J := \emptyset$,}\\
  {Let $K = \{i\in[n]\;|\; y_i=3\}$}\\
  \For{$i\in K$}{
  {Let $j_i$ be a randomly and uniformly chosen element of $[2:n]\setminus J$}\\
  {$J := J \cup \{j_i\}$}
  }
  \For{$t\in[n]$}{
  {$y_t' = g(y_t, \decompprob{\Aldp[t](z_{1:t-1})}, q_{\Aldp[t](z_{1:t-1})})$}\\
  {Sample $z_t$ from $ \begin{cases} {\mathcal{Q}_1^0}^{(t)}(z_{1:t-1}) & \text{if } y_i'=0; \\
                                    {\mathcal{Q}_1^1}^{(t)}(z_{1:t-1}) & \text{if } y_i'=1;  \\
                                    \mathcal{Q}_{1}^{(t)}(z_{1:t-1}) & \text{if } y_i'=2;  \\
                                    \mathcal{Q}_{j_t}^{(t)}(z_{1:t-1}) & \text{if } y_i'=3. \end{cases}$}
  
  }
  {\textbf{return} $z_1, \ldots, z_n$}
\end{algorithm}

\begin{proof}
Let $X_0=\{x_1^0,x_2,\cdots,x_n\}$ and $X_1=\{x_1^1,x_2,\cdots,x_n\}$ be two neighbouring datasets in $\mathcal{D}^n$. 
A description of the post-processing function is given in Algorithm~\ref{postprocessing2}. We claim that for $b\in\{0,1\}$, $f(\threedistqb{\eps_0}{p^*}{q^*})\disteq \shuffler(X_b)$. The key observation is that by assumption we have decompositions such that for all $t\in[n]$ and $z_{1:t-1}$, there exists $(\decompprob{\Aldp[t](z_{1:t-1})}, q_{\Aldp[t](z_{1:t-1})})$ such that $\decompprob{\Aldp[t](z_{1:t-1})}\le p^*$ and $q_{\Aldp[t](z_{1:t-1})}\ge q^*+2(p^*-\decompprob{\Aldp[t](z_{1:t-1})})$ and letting $p_t:=\decompprob{\Aldp[t](z_{1:t-1})}$ and $q_t:=q_{\Aldp[t](z_{1:t-1})}$:
\begin{align*}
\forall b\in\{0,1\}, \;\;\mathcal{R}^{(t)}(z_{1:t-1}, x_1^b) &= e^{\eps}p_t{\mathcal{Q}_1^b}^{(t)}(z_{1:t-1})+p_t{\mathcal{Q}_1^{1-b}}^{(t)}(z_{1:i-1}) + (1-p_t-e^{\eps_0}p_t)\mathcal{Q}_1^{(t)}(z_{1:i-1}),
\end{align*}

\begin{align*}\forall i\in[2,n], \;\;\mathcal{R}^{(t)}(z_{1:t-1}, x_i) &= p_t{\mathcal{Q}_1^0}^{(t)}(z_{1:i-1})+p_t{\mathcal{Q}_1^1}^{(t)}(z_{1:i-1})
+q_t\mathcal{Q}_1^{(t)}(z_{1:i-1})+(1-q_t-2p_t)\mathcal{Q}_i^{(t)}(z_{1:i-1}).
\end{align*}

The mixture coefficients of the random variables $\Aldp[t]$ do not necessarily match those of ${Y_{1,p^*}^b}$, and $Y_{p^*,q^*}$. However, for any $p<p^*$ and $q>q^*+2(p^*-p)$ we can define a post-processing function   $g(\cdot,p,q)$ such that $g({Y_{1,p^*}^b}, p,q)= {Y_{1,p}^b}$, and $g(Y_{p^*,q^*}, p, q)=Y_{p,q}$. This function is given by $g(0,p,q)=0$ with probability $p/p^*$ and 2 otherwise, $g(1,p,q)=1$ with probability $p/p^*$ and 2 otherwise, $g(2,p,q)=2$ with probability 1, and $g(3,p,q)=3$ with probability $(1-q-2p)/(1-q^*-2p^*)$ and 2 otherwise.

Let $y\in\{0,1,2,3\}^n$ be a permutation of the local reports given by 1 copy of $Y_{1,p^*}^b$ and $n-1$ copies of $Y_{p^*,q^*}$; this is equivalent to a sample from $\threedistqb{\eps_0}{p^*}{q^*}$. Given the hidden permutation $\pi$, we can generate a sample from $\shuffler(X_b)$ by sequentially transforming $y_t'=g(y_t, \decompprob{\Aldp[t](z_{1:t-1})}, q_{\Aldp[t](z_{1:t-1})})$ to obtain the correct mixture components, then sampling from the corresponding mixture component. The difficulty then lies in the fact that conditioned on a particular instantiation $y=v$, the permutation $\pi|_{y=v}$ is not independent of $b$.

The first thing to note is that if $v_t=0, 1$ or $2$, then the corresponding mixture component ${\mathcal{Q}_1^0}^{(t)}(z_{1:t-1}), {\mathcal{Q}_1^1}^{(t)}(z_{1:t-1})$ or ${\mathcal{Q}_1}^{(t)}(z_{1:t-1})$ (respectively), is independent of $\pi$. Therefore, in order to do the appropriate post-processing, it suffices to know the permutation $\pi$ restricted to the set of users who sampled $3$, $K=\pi(\{i: y_i = 3\})$. The set $K$ of users who select 3 is independent of $b$ since $Y_{1,p^*}^0$ and $Y_{1,p^*}^1$ both have zero probability of sampling $3$. The probability of being included in $K$ is identical for each $i\in[2 \colon n]$, and
any ordering of the users in $K$ is equally likely, so we can choose a random assignment. Now that we have an assignment for $K$, we can sample from the correct mixture components as desired.
\end{proof}

\subsubsection{$k$-randomized response}
Lemma~\ref{shuffletobinoms2new} allows us give an upper bound for \texttt{kRR}.

\begin{theorem}[Replacement for Theorem 5.2]\label{kRRthmnew} For a domain $\mathcal{D}$, let $\Aldp[i]:\out[1]\times\cdots\times\out[i-1]\times\mathcal{D}\to\out[i]$ for $i\in[n]$ (where $\out[i]$ is the range space of $\Aldp[i]$) be a sequence of algorithms such that $\Aldp[i](z_{1:i-1}, \cdot)$ is a $\eps_0$-DP local randomizer for all values of auxiliary inputs $z_{1,i-1}\in\out[1]\times\cdots\times\out[i-1]$. Let $\shuffler:\mathcal{D}^n\to\out[1]\times\cdots\times \out[n]$ be the algorithm that given a dataset $x_{1:n}\in\mathcal{D}^n$, samples a uniformly random permutation $\pi$, then sequentially computes $z_i=\Aldp[i](z_{1:i-1}, x_{\pi(i)})$ for $i\in[n]$ and outputs $z_{1:n}$. Assume that for some $k\in\mathbb{N}$ we have that for all $i\in[n]$, there exists a deterministic function $f^{(i)}:\out[1]\times\cdots\times\out[i-1]\times\mathcal{D}\to[k]$ such that $\Aldp[i](z_{1:i-1}, x)=\kRR(f^{(i)}(z_{1:i-1},x))$. 
Then for any distance measure $D$ that satisfies the data processing inequality, and any two neighbouring datasets $X_0$ and $X_1$, \[D(\shuffler(X_0), \shuffler(X_1))\le D\left(\threedistqone{\eps_0}{\frac{1}{e^{\eps_0}+k-1}}{\frac{k-2}{e^{\eps_0}+k-1}}, \threedistqtwo{\eps_0}{\frac{1}{e^{\eps_0}+k-1}}{\frac{k-2}{e^{\eps_0}+k-1}}\right)\]
\end{theorem}
\begin{proof}
Let $X_0$ and $X_1$ be neighbouring datasets. Given $i$ and $z_{1:i-1}$, if $\Aldp[i](z_{1:i-1}, x)=\kRR(f^{(i)}(z_{1:i-1},x))$, let $\mathcal{U}(x)$ be the random variable that outputs $f^{(i)}(z_{1:i-1},x)$ with probability 1, $\mathcal{U}$ be the uniform distribution over $\mathcal{D}\backslash \{f^{(i)}(z_{1:i-1},x_1^0), f^{(i)}(z_{1:i-1},x_1^1)\}$.
Then, letting $p=\frac{1}{e^{\eps_0}+k-1}$ and $q=\frac{k-2}{e^{\eps_0}+k-1}$
\begin{align*}
\Aldp[i](z_{1:i-1}, x_1^b) &= e^{\eps_0}p\mathcal{U}(x_0^b)+p\mathcal{U}(x_0^{1-b})+(1-p-e^{\eps_0}p))\mathcal{U}\\
\forall i\in[2:n], \Aldp[i](z_{1:i-1}, x_i) &= p\mathcal{U}(x_0^b)+p\mathcal{U}(x_0^{1-b})+q\mathcal{U}+(1-q-2p)\mathcal{U}(x_i)
\end{align*}
Hence by Lemma~\ref{shuffletobinoms2new} and the data processing inequality, $D(\shuffler(X_0),\shuffler(X_1))\le D(\threedistqone{\eps_0}{p}{q}, \threedistqtwo{\eps_0}{p}{q})$ as required.
\end{proof}

In Figure~\ref{kRRfig} we can see that the improved bounds for $\kRR$ are indeed tighter than the general bound (for all $k$, $\texttt{kRR}\in\extremalset{\eps_0}$), and the privacy guarantee on the output of the shuffler improves as $k$ increases, as expected. We also compare to the best known bound from prior work \cite{Balle:2019}. We can see that the upper and lower bounds presented in Theorem~\ref{kRRlowerbound} and Theorem~\ref{kRRthm} are indeed numerically close.

\begin{figure}
    \centering
    \includegraphics[scale=0.7]{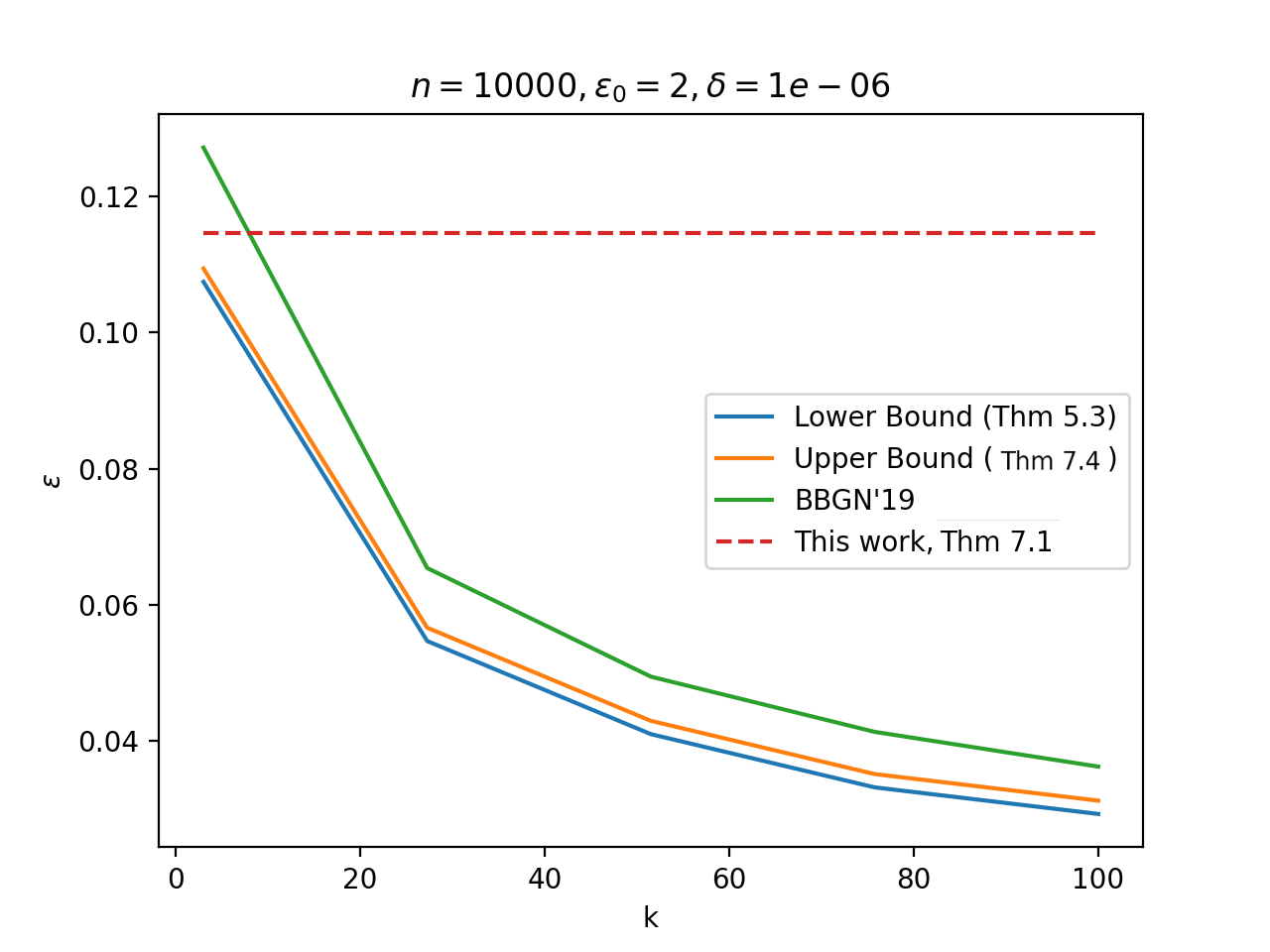}
    \caption{Privacy amplification by shuffling for $\kRR$ as a function of $k$ using bound from Theorem~\ref{kRRthm}. The horizontal dashed line corresponds to the general upper bound from Theorem~\ref{mainanalyticalthm}. }
    \label{kRRfig}
\end{figure}

\newpage

\printbibliography
\newpage

\appendix
\section{Proofs for Section~\ref{generalamplificationreduction}}\label{appendixreduction}

\begin{lemma}\label{hypercube}
Let $A=[1,e^{\eps}]^k$ and $B=\{1,e^{\eps}\}^k$. Every vector in $A$ can be written as a convex combination of vectors in $B$.
\end{lemma}

\begin{lemma}\label{sizetwopdf-app}\cite[Lemma IV.4]{ye2018optimal}
If $\mathcal{R}:\mathcal{D}\to\mathcal{S}$ is an $\eps$-DP local randomizer, and both $\mathcal{D}$ and $\mathcal{S}$ are finite, then there exists a finite output space $\mathcal{Z}$, a local randomizer $\mathcal{R}':\mathcal{D}\to\mathcal{Z}$, and a post-processing function $\Phi:\mathcal{Z}\to\mathcal{S}$ such that $\Phi\circ\mathcal{R'}=\mathcal{R}$ and for all $z\in\mathcal{S}$, there is a $p_z\in[0,1/(e^{\eps_0}+1)]$ so that every $x\in\mathcal{D}$ satisfies $\Pr(\mathcal{R}'(x)=z)\in\{p_z, e^{\eps_0}p_z\}$.
\end{lemma}

\begin{proof}[Proof of Lemma~\ref{sizetwopdf-app}]
Let $\mathcal{D}=[k]$ and $\mathcal{Z}=\{1,e^{\eps}\}^k$ so we can use the elements of $\mathcal{D}$ to index the coordinates of $\mathcal{Z}$. For all $s\in\mathcal{S}$, let $p_s=\min\{\Pr(\mathcal{R}(x)=s)\;|\; x\in\mathcal{D}\}$. Differential privacy implies that the vector $(1/p_s)[\Pr(\mathcal{R}(x_1), \cdots, ,\mathcal{R}(x_k)]\in[1,e^{\eps}]^k$.
By Lemma~\ref{hypercube}, there exists $\{\lambda_{s,z}\;|\; s\in\mathcal{S}, z\in\mathcal{Z}\}$ such that $\sum_{z\in\mathcal{Z}}\lambda_{s,z}=1$ and \begin{align*}
\frac{1}{p_s}\Pr(\mathcal{R}(x)=s)&=\sum_{z\in\mathcal{Z}} \lambda_{s,z} z_x.
\end{align*}
Let $p_z=\sum_{s\in\mathcal{S}}p_s\lambda_{s,z}$ and write
\[\Pr(\mathcal{R}(x)=s)=\sum_{z\in\mathcal{Z}} \left(\frac{p_s\lambda_{s,z}}{p_z}\right) p_z z_x.\]
Define $\mathcal{R'}:\mathcal{D}\to\mathcal{Z}$ by $\Pr(\mathcal{R'}(x)=z)=p_z z_x$. This is well-defined since \[\sum_{z\in\mathcal{Z}}p_z z_x = \sum_{z\in\mathcal{Z}}\sum_{s\in\mathcal{S}}p_s\lambda_{s,z}z_x = \sum_{s\in\mathcal{S}}p_s\sum_{z\in\mathcal{Z}}\lambda_{s,z}z_x = \sum_{s\in\mathcal{S}}\Pr(\mathcal{R}(x)=s)=1.\]
Define $\Pr(\Phi(z)=s)=\frac{p_s\lambda_{s,z}}{p_z}$. Note that $\phi$ is well-defined since for every $z\in\mathcal{Z}$, $\sum_{s\in\mathcal{S}}\frac{p_s\lambda_{s,z}}{p_z}=1$, by the definition of $p_z$. Therefore, $\Phi\circ\mathcal{R}'=\mathcal{R}$. Finally, since each $z\in\{1,e^{\eps}\}^k$, we have the final claim that $\Pr(\mathcal{R'}(x)=z)\in\{p_z, e^{\eps_0}p_z\}.$
\end{proof}

The proof of Theorem~\ref{thm:shuffling4adp} relies on the following lemma from \citep{FeldmanMT:2020}.

\begin{lemma}\cite[Lemma A.1]{FeldmanMT:2020}\label{binomials-aux} Let $p\in(0,1/2]$, $\delta >0$, and $n\in\mathbb{N}$ be such that $n~\ge~ \frac{8\ln(2/\delta)}{p}$, and let \[\eps=\ln\left(1+\frac{\sqrt{32\ln(4/\delta)}}{\sqrt{pn}}+\frac{4}{pn}\right).\]
Consider the process where we sample $C\sim \bin(n-1, 2p)$ and $A\sim\bin(C, 1/2)$. Let $P=(A+1,C-A)$ and $Q=(A,C-A+1)$, then
\[\Pr_{(a,c)\sim P}\left(-\eps\le\ln\frac{\Pr(P=(a,c))}{\Pr(Q=(a,c))}\le \eps\right)\ge 1-\delta \text{ and } \Pr_{(a,c)\sim Q}\left(-\eps\le\ln\frac{\Pr(P=(a,c))}{\Pr(Q=(a,c))}\le \eps\right)\ge 1-\delta.\]
In particular, $P$ and $Q$ are $(\eps, \delta)$-indistinguishable.
\end{lemma}

The proof of Theorem~\ref{mainanalyticalthm} is directly implied by the following lemma 
with $p = 2/(e^{\eps_0}+1)$ which proves a slightly stronger statement that we use in the proof of Corollary~\ref{RDPbound}.
\begin{lemma}
\label{lem:tail-bound}
Let $\eps_0\ge 0$, $\delta >0$, $p\in(0,1/(e^{\eps_0}+1)]$, and $n\in\mathbb{N}$ be such that $n~\ge~ \frac{8\ln(2/\delta)}{p}$. 
	Let $C\sim \bin(n-1, 2p)$, $A\sim\bin(C, 1/2)$ and $\Delta \sim \Ber\left(\frac{e^{\eps_0}}{e^{\eps_0}+1}\right)$.  Let $P=(A+\Delta, C-A+1-\Delta)$ and $Q=(A+1-\Delta,C-A+\Delta)$. 
For \[\eps = \ln
\left(1+\frac{e^{\eps_0}-1}{e^{\eps_0}+1}\left(\frac{\sqrt{32\ln(4/\delta)}}{\sqrt{pn}}+\frac{4}{pn}\right)\right)\] we have that
\[\Pr_{(a,c)\sim P}\left(-\eps\le\ln\frac{\Pr(P=(a,c))}{\Pr(Q=(a,c))}\le \eps\right)\ge 1-\delta \text{ and }\]
\[\Pr_{(a,c)\sim Q}\left(-\eps\le\ln\frac{\Pr(P=(a,c))}{\Pr(Q=(a,c))}\le \eps\right)\ge 1-\delta .\]
In particular, $P$ and $Q$ are $(\eps, \delta)$- indistinguishable.
\end{lemma}

\begin{proof}
	Consider the process where we sample $C\sim \bin(n-1, e^{-\eps_0})$ and $A\sim\bin(C, 1/2)$. Let $P_0=(A+1,C-A)$ and $Q_0=(A,C-A+1)$ then according to Lemma~\ref{binomials-aux},
\[\Pr_{(a,c)\sim P_0}\left(-\eps'\le\ln\frac{\Pr(P_0=(a,c))}{\Pr(Q_0=(a,c))}\le \eps'\right)\ge 1-\delta \text{.   and    } \Pr_{(a,c)\sim Q_0}\left(-\eps'\le\ln\frac{\Pr(P_0=(a,c))}{\Pr(Q_0=(a,c))}\le \eps'\right)\ge 1-\delta,\]
where $\eps'=\ln\left(1+\frac{\sqrt{32\ln(4/\delta)}}{\sqrt{pn}}+\frac{4}{pn}\right).$
Now, let $\alpha=e^{\eps_0}/(e^{\eps_0}+1)$ note that $P = \alpha P_0 +(1-\alpha) Q_0$ and $Q = (1-\alpha)P_0+\alpha Q_0.$
Let $f(z)=\frac{\alpha z + (1-\alpha)}{(1-\alpha)z+\alpha}$, then $f'(z) = \frac{2\alpha-1}{((1-\alpha)z+\alpha)^2}\ge 0$. Therefore, $\max_{z\in[e^{-\eps'},e^{\eps'}]}f(z) = f(e^{\eps'})$. Therefore, if $\frac{\Pr(P_0=(a,c))}{\Pr(Q_0=(a,c))}\in[e^{-\eps'},e^{\eps'}]$ then
\begin{align*}
	\frac{\Pr(P=(a,c))}{\Pr(Q=(a,c))}   &= \frac{\alpha\frac{\Pr(P_0=(a,c))}{\Pr(Q_0=(a,c))}+(1-\alpha)}{(1-\alpha)\frac{\Pr(P_0=(a,c))}{\Pr(Q_0=(a,c))}+\alpha}\\
	&\le \frac{\alpha e^{\eps'}+(1-\alpha)}{(1-\alpha)e^{\eps'}+\alpha}\\
	&= 1+\frac{(2\alpha-1)(e^{\eps'}-1)}{(1-\alpha)e^{\eps'}+\alpha}\\
	&\le 1+(2\alpha-1)(e^{\eps'}-1)\\
    &\le 1+\frac{e^{\eps_0}-1}{e^{\eps_0}+1}\left(\frac{\sqrt{32\ln(4/\delta)}}{\sqrt{pn}}+\frac{4}{pn}\right)
\end{align*}
and
\begin{align*}
	\frac{\Pr(P=(a,c))}{\Pr(Q=(a,c))}&\ge \frac{\alpha e^{-\eps'}+(1-\alpha)}{(1-\alpha)e^{-\eps'}+\alpha}\\
	&=\frac{(1-\alpha)e^{\eps'}+\alpha}{\alpha e^{\eps'}+(1-\alpha)}\\
&\ge \frac{1}{1+\frac{e^{\eps_0}-1}{e^{\eps_0}+1}\left(\frac{\sqrt{32\ln(4/\delta)}}{\sqrt{pn}}+\frac{4}{pn}\right)}
\end{align*}

\end{proof}

\section{Proof of Corollary~\ref{RDPbound-corrected}}
\label{app:rdp-cor}

To prove Corollary~\ref{RDPbound-corrected}  we use the main reduction in \cite{FeldmanMT:2020} together with the slight 
strengthening of their approximate DP  bound that directly controls the tail of the privacy loss random variable that we gave in Lemma~\ref{lem:tail-bound}.

\begin{theorem}\cite[Theorem 3.2]{FeldmanMT:2020}\label{mainanalyticalthm-previous}
	For a domain $\mathcal{D}$, let $\Aldp[i]:\out[1]\times\cdots\times\out[i-1]\times\mathcal{D}\to\out[i]$ for $i\in[n]$ (where $\out[i]$ is the range space of $\Aldp[i]$) be a sequence of algorithms such that $\Aldp[i](z_{1:i-1}, \cdot)$ is an $\eps_0$-DP local randomizer for all values of auxiliary inputs $z_{1:i-1}\in\out[1]\times\cdots\times\out[i-1]$. Let $\shuffler:\mathcal{D}^n\to \out[1]\times\cdots\times \out[n]$ be the algorithm that given a dataset $x_{1:n}\in\mathcal{D}^n$, samples a permutation $\pi$ uniformly at random, then sequentially computes $z_i=\Aldp[i](z_{1:i-1}, x_{\pi(i)})$ for $i\in[n]$ and outputs $z_{1:n}$. Let $X_0$ and $X_1$ be two arbitrary neighboring datasets in $\mathcal{D}^n$.  Let $C\sim \bin(n-1, e^{-\eps_0})$, $A\sim\bin(C, 1/2)$ and $\Delta \sim \Ber\left(\frac{e^{\eps_0}}{e^{\eps_0}+1}\right)$. 
	Then for any distance measure $D$ that satisfies the data processing inequality,
	\[D(\shuffler(X_0)\|\shuffler(X_1))\le D\left((A+\Delta, C-A+1-\Delta)\Big\|(A+1-\Delta,C-A+\Delta)\right).\]
\end{theorem}

\begin{proof}[Proof of Corollary~\ref{RDPbound-corrected}]
We can, for brevity, restrict ourselves to $\eps_0 \geq 1$ since the result is implied by Erlingsson et al.~\cite{ErlingssonFMRTT19} for $\eps_0 \leq 1$. Now let $P$ and $Q$ be the distributions defined in Theorem~\ref{mainanalyticalthm-previous}. Namely, for $C\sim \bin(n-1, e^{-\eps_0})$, $A\sim\bin(C, 1/2)$ and $\Delta \sim \Ber\left(\frac{e^{\eps_0}}{e^{\eps_0}+1}\right)$,  let $P=(A+\Delta, C-A+1-\Delta)$ and $Q=(A+1-\Delta,C-A+\Delta)$. 
Then, by Lemma~\ref{lem:tail-bound} with $p=1/(2e^{\eps_0})$, we have that, for $n~\ge~ \frac{16\ln(2/\delta)}{e^{\eps_0}}$ and $\eps = \ln
\left(1+\frac{e^{\eps_0}-1}{e^{\eps_0}+1}\left(\frac{8\sqrt{e^{\eps_0}\ln(4/\delta)}}{\sqrt{n}}+\frac{8e^{\eps_0}}{n}\right)\right)$, it holds that
\[\Pr_{(a,c)\sim Q}\left[e^{-\eps}\le \frac{\Pr[P=(a,c)]}{\Pr[Q=(a,c)]}\le e^\eps\right]\ge 1-\delta .\]
Note that for $n~\ge~ \frac{16\ln(2/\delta)}{e^{\eps_0}}$ and $\delta \leq 1$, \[ \eps \leq  \frac{e^{\eps_0}-1}{e^{\eps_0}+1}\left(\frac{8\sqrt{e^{\eps_0}\ln(4/\delta)}}{\sqrt{n}}+\frac{8e^{\eps_0}}{n}\right)\leq \frac{16 \sqrt{e^{\eps_0}}}{\sqrt{n}} \cdot \sqrt{\ln(2/\delta)} .\]
The condition on $n$ is equivalent to $\delta \geq 2 e^{-\frac{n}{16 e^{\eps_0}}}$. This implies that for $\sigma = 16 \sqrt{\frac{e^{\eps_0}}{n}}$ we have that for any $\delta \geq e^{-\frac{n}{16 e^{\eps_0}}}$,
\[\Pr_{(a,c)\sim Q}\left[\left| \ln \frac{\Pr[P=(a,c)]}{\Pr[Q=(a,c)]} \right|\ge \sigma \sqrt{\ln(1/\delta)}  \right]\leq 2 \delta ,\]
To ensure that the condition $\delta_{\min} = e^{-\frac{n}{16 e^{\eps_0}}} \leq e^{-\alpha \eps_0} \cdot \alpha^2\sigma^2 /4$ is satisfied for $\alpha > 1$ it suffices to take \[\alpha \leq \frac{1}{\eps_0} \left(\frac{n}{16e^{\eps_0}} - \ln \left(\frac{n}{64 e^{\eps_0}}\right) \right) .\]
In particular, it is satisfied for $\alpha \leq \frac{n}{32 \eps_0 e^{\eps_0}}$.
This means that we can apply Theorem~\ref{thm:adp2rdp} to conclude that for $1 \leq  \alpha \leq \frac{n}{32 \eps_0 e^{\eps_0}}$ , we have that \[D^\alpha(P\|Q) =O(\alpha \sigma^2) = O(\alpha e^{\eps_0}/n).\]

\end{proof}

\section{Proof of Lemma~\ref{maxdivergence}}

\begin{proof}[Proof of Lemma~\ref{maxdivergence}]
Define a post-processing function $g:\mathbb{N}^2\to\mathbb{N}^2$ by $g(d,e)=(\bin(d,\frac{p}{p'}), \bin(e,\frac{p}{p'}))$. 
We claim that $g(\threedistone{\eps_0}{p'})\disteq \threedistone{\eps_0}{p}$ and $g(\threedisttwo{\eps_0}{p'}) \disteq \threedisttwo{\eps_0}{p}$.
Let $\gamma=p/p'$. Consider three random variables $X(p)$, $Y_1(p)$ and $Y_2(p)$ with outputs in 
$\{(1,0),(0,1),(0,0)\}$. Suppose that\[\Pr(X(p)=(1,0))=\Pr(X(p)=(0,1))=p\;\;\;\text{and}\;\;\;\Pr(X(p)=(0,0))=1-2p,\]
\[\Pr(Y_1(p)=(1,0))=p,\;\Pr(Y_1(p)=(0,1))=e^{\eps}p\;\;\;\text{and}\;\;\;\Pr(Y_1(p)=(0,0))=1-p-e^{\eps}p,\]
and
\[\Pr(Y_2(p)=(1,0))=e^{\eps}p,\;\Pr(Y_2(p)=(0,1))=p\;\;\;\text{and}\;\;\;\Pr(Y_2(p)=(0,0))=1-p-e^{\eps}p.\]
Note that $g(X(p'))\disteq X(p)$, $g(Y_1(p'))\disteq Y_1(p)$ and $g(Y_2(p'))\disteq Y_2(p)$. Also, if $X_1, \cdots, X_{n-1}$ are $n-1$ independent copies of $X$ then
\begin{align*}
g(\threedistone{\eps_0}{p'})&\disteq g(\sum_{i=1}^{n-1} X_i(p') + Y_1(p'))\\
&\disteq \sum_{i=1}^{n-1} g(X_i(p'))+g(Y_1(p'))\\
&\disteq \sum_{i=1}^{n-1} X_i(p)+Y_1(p)\\
&\disteq \threedistone{\eps_0}{p}
\end{align*}
Similarly, $g(\threedisttwo{\eps_0}{p'})\disteq\threedisttwo{\eps_0}{p}$. Therefore, the lemma follows from the post-processing inequality (Lemma~\ref{postprocess}).
\end{proof}

\end{document}